\documentclass[12pt,titlepage]{article}
\usepackage[centertags]{amsmath}
\usepackage{amsfonts}
\usepackage{amssymb}
\usepackage{amsthm}
\usepackage{epsfig}
\usepackage{newlfont}
\usepackage{verbatim}
\newtheorem{thm}{Theorem}

\newtheorem{prp}{Proposition}

\newtheorem{example}{Example}

\newtheorem{lemma}{Lemma}
\newtheorem{lemy}{Lemma}[lemma]

\newtheorem{corollary}{Corollary}

\newtheorem{definition}{Definition}
\newtheorem{defn}{Definition}[prp]
\newcommand{\bcor}{\begin{corollary}}
\newcommand{\ecor}{\end{corollary}}

\def\real{\hbox{\rm\setbox1=\hbox{I}\copy1\kern-.45\wd1 R}}
\def\neal{\hbox{\rm\setbox1=\hbox{I}\copy1\kern-.45\wd1 N}}

\renewcommand{\thefootnote}{\fnsymbol{footnote}}

\title{Ex-Post Equilibrium and VCG Mechanisms\thanks{This work is partly based on Rozen's
M.Sc thesis done under the supervision of Ron Holzman}}
\author{Rakefet Rozen \\
Rann Smorodinsky
\thanks{Davidson Faculty of Industrial Engineering and Management,
Technion, Haifa 32000, Israel. Tel:$\ +972\ (0)4 \ 8294422$, Fax:$\
+972 \ (0)4 \ 8295688$, $<${\tt rann@ie.technion.ac.il}$>$. The
hospitality of Department of Economics and Business, Pompeu Fabra
University, Barcelona, Spain and financial support from the Gordon Center for System Engineering at the Technion and the Sapir center are gratefully acknowledged. } }

\begin{document}

\begin{titlepage}

\pagestyle{empty}

\maketitle

\begin{abstract}

Consider an abstract social choice setting with incomplete
information, where the number of alternatives is large. Albeit
natural, implementing VCG mechanisms is infeasible due to the
prohibitive communication constraints. However, if players restrict
attention to a subset of the alternatives, feasibility may be
recovered.

This paper characterizes the class of subsets which induce an
ex-post equilibrium in the original game. It turns out that a
crucial condition for such subsets to exist is the existence of a
type-independent optimal social alternative, for each player. We
further analyze the welfare implications of these restrictions.

This work follows work by Holzman, Kfir-Dahav, Monderer and
Tennenholtz \cite{RNDM} and Holzman and Monderer \cite{HolMon2002}
where similar analysis is done for combinatorial auctions.


\end{abstract}

\end{titlepage}

\renewcommand{\thefootnote}{\arabic{footnote}}
\setcounter{footnote}{0}

\section{Introduction}

A classical result in the theory of mechanism design, known as the
Vickrey-Clarke-Groves (VCG) mechanism, asserts that if agents'
utility is based on his valuations of goods and his monetary
situations, and these two are quasi-linear, then it is possible to
align all agents' interest such that all agents would want to
maximize the social welfare. Consequently, for each agent,
truthful reporting dominates any other strategy. This is true even
if an agent does not know exactly how many other agents are
participating.

This remarkable result has been a cornerstone of many design
problems, such as taxation, public good, auctions and more. In
recent years, two major practical difficulties have been
identified with this approach. Both are related to the design of
``large" problems. Namely, problems where the number of social
alternatives, $M$, is big.

The first problem relates to the {\bf computational complexity} of
determining the social maxima, among all $M$ alternatives. The
second problem refers to the {\bf communication complexity} of each
agent's message. For an agent to fully report her valuation she must
provide $M$ numbers, and if $M$ is large this may be quite
prohibitive. Although, these problems are of different nature both
stem from the large size of social alternatives.

The classical example which gives rise to such problems is that of combinatorial auctions:

\begin{example}
\label{ex1}
Assume $N$ agents bid for $G$ different goods and each agent cares only about the goods she receives. An agent's type is a vector of valuations, one for each of the $2^G-1$ nonempty
subsets of $G$. The case $|G|=50$ is already prohibitively large.
\end{example}

This example has been studied extensively but it is far  from being the unique setting giving rise communication complexity. Some other examples are:

\begin{example}
\label{ex1}
Assume $N$ agents bid for $G$ different goods and each agent cares for the partition of the goods. Namely agent $i$ utility also depends on the good agent $j$ receives. The type of an agent is a valuation for each of $(N+1)^G$ possible allocations. The case $|G|=|N|=20$ is prohibitively large.
\end{example}

Another example is that of ordering:

\begin{example}
Assume agents have to decide on an allocation of $C$ candidates to
$C$ positions (or an order of $C$ jobs to be completed serially).
Any agent then assigns a value to any of the $C!$
possible allocations. Once again, for $C=50$ this becomes too
demanding.
\end{example}

Additional examples are location problems where $F$ facilities need to be located in $L$ possible locations (inducing $|L|^{|F|}$ possibilities), network building where a set $E$ edges needs to be built in order to connect $V$, and more.

The problem of finding the social optimum has been studied, from a
computational complexity point of view by various authors. Some
examples are Rothkopf et al \cite{RothPekHar}, Fujishima et al
\cite{FujiBrownShoham}, Anderson et al \cite{Anderson}, Sandholm et
al \cite{Sandholm01} and Hoos and Boutilier \cite{HoosBoutilier}.
The communication complexity problem motivated researchers to
characterize alternative mechanisms which involve less demanding
strategies than truthfulness. Some examples for this line of
research in an auction setting are Gul and Stacchetti
\cite{GulSta00}, Parkes \cite{Parkes99}, Parkes and Ungar
\cite{ParkesUngar}, Wellman et al \cite{Wellmangeb} and Bikhchandani
et al \cite{Bikhchandani2001}.

In this paper we follow the framework of Holzman, Kfir-Dahav,
Monderer and Tennenholtz \cite{RNDM} and Holzman and Monderer
\cite{HolMon2002}, who do not look for alternative mechanisms but
rather study an alternative solution concept for the class of VCG
mechanisms. In particular these two papers study the class of ex post equilibria of VCG mechanisms which exhibit many of the properties of dominant strategies, yet allow for less demanding communication complexity. In Holzman, Kfir-Dahav, Monderer
and Tennenholtz \cite{RNDM} a family of ex post equilibria, called
bundling equilibria, is introduced and its efficiency is studied.
A bundling strategy is one where agents report valuations
on a subset of all alternatives, where the subset must contain the
empty set, be close under union of mutually exclusive sets, and
close under complementarities.  Holzman and Monderer \cite{HolMon2002} show that this family exhausts all the ex-post equilibria in combinatorial auction settings.

Whereas these two papers focus on the setting of example \ref{ex1}, namely on combinatorial auction without externalities, we study similar issues in a general and more abstract setting, making our results relevant to all the aforementioned examples. Additionally, the results obtained by in the two papers are restricted to the case of $N>2$ players while we our results encompass the $2$ player case as well.
A more detailed comparison of our results with the results in the combinatorial auction setting is provided in section 5.

This paper focuses on three research questions: existence, characterization and efficiency.
In particular we show that for the most general case, where valuations are arbitrary, the unique ex
post equilibrium is that of the weakly dominant strategies. However,
by imposing one restriction on the valuation space, namely that each
agent´s optimal choice is independent of his valuation, we can
recover the positive result and provide a large set of ex post
equilibria. In terms of efficiency loss, we show that in the general
case the efficiency loss grows with the number of players, and
cannot be bounded uniformly. In fact we show that the efficiency
loss is in the order of magnitude of the number of players and this
cannot be improved upon. However, we propose two types of
restrictions on the set of valuations which induce a uniform bound
on the efficiency loss.

A related strand of the literature is that on ex-post
implementation, which is part of the mechanism design literature.
The research goal of most papers on ex-post implementation is to
characterize the conditions needed for obtaining a mechanism which
implements some social choice functions under the ex-post
equilibrium solution. Some recent contributions to this are
Bergemann and Morris \cite{BergemannMorris2008}, Bikhchandani
\cite{Bikhchandani2006} and Jehiel et al \cite{JMMZ}. Needless to
mention, given the research goal of this literature, that these
papers do not yield mechanisms which have low complexity. Nisan and
Segal \cite{NisSeg02} study the tradeoff between communication
complexity and efficiency in allocation problems.

In section 2 we provide the basic model and definitions. Section 3
discusses existence and structure of ex-post equilibria, whereas section 4 analyzes their efficiency.
Section 5 is devoted to a discussion of the results for the combinatorial setting and a comparison with the results of Holzman, Kfir-Dahav, Monderer
and Tennenholtz \cite{RNDM} and Holzman and Monderer \cite{HolMon2002} . The proofs are relegated to the appendix.

\section{Model}

Let $A$ be a finite abstract set of social alternatives and let $N$
be a finite set of $n=|N|$ agents. A valuation function for agent
$i$ is a function $v_i:A \to \real$ and $v=(v_1,\ldots,v_n) \in
(\real^A)^n$ denotes a vector of valuations, one for each agent.

An {\it allocation mechanism}, $M: (\real^A)^n \to A$, chooses a
single alternative for each vector of valuations. We say that $M$ is
a social welfare maximizer if $M(v) \in \arg\max_{a\in A} \sum_i
v_i(a)$ for all $v \in (\real^A)^n$. Let ${\cal M}_N$ be the set of
all social welfare maximizers for the set of agents $N$. Note that
two elements in ${\cal M}_N$ differ only in the way they break ties.

A set of agents $N$, a social welfare maximizing mechanism, $M$,
and a set of functions $h_i:({\real^A})^{N - \{i\}} \to \real$,
$\forall i \in N$, and a set of valuations, ${\cal V}_i \subset
\real^A$, for each $i$, defines a {\it Vickrey-Clarke-Groves (VCG)
game}, denoted $\Gamma(N,M,h,{\cal V})$, where $h=(h_i)_{i=1}^n$
and ${\cal V} = \prod_{i\in N} {\cal V}_i$
\begin{itemize}
\item Agent $i$'s strategy, $b_i:{\cal V}_i \to \real^A$, maps his
true valuation to an announced valuation (possibly not in ${\cal
V}_i$). $b_i(v)(a)$ is the valuation announced for alternative $a$,
when the actual valuation is $v$ and the strategy is $b_i$. Let
$b=(b_1,\ldots,b_n)$ denote the agents' strategy profile and let
$b_i({\cal V}_i) \subset \real^A$ be the set
of all possible announcements of $i$.%
\item For any $v \in {\cal V}$ and any strategy profile, $b$, the
utility of agent $i$ is $U_i=v_i(M(b(v))) + \sum_{j\not =
i}b_j(v_j)(M(b(v))) - h_i(b_{-i}(v_{-i}))$.
\end{itemize}

We refer to the set of all VCG games where $h_i=0$ as {\it simple
VCG games}.

We note two observations about such VCG games:
\begin{itemize}
\item The bid $b_i$ is a best response for agent $i$, against $b_{-i}$,
in some VCG game if and only if it is a best response in all VCG
games. \item The truth-telling strategy, $b_i(v_i) = v_i$, weakly
dominates any other strategy in any VCG game. Furthermore, it is the
unique strategy with that property (up to a constant).
\end{itemize}

Given the above comments it seems that there is no need for any
additional game theoretic analysis of VCG games. However, recently
there has been a growing interest in the literature in situations
where agents face a large number of social alternatives, in which
case the communication of an agent's valuation is prohibitively
long, and practical reasons render the truth-telling strategy as
impossible. One example for such a situation is a combinatorial
auctions, where the number of social alternatives is exponential in
the number of goods. Another example is an assignment problem, where
"jobs" are assigned to "resources" (e.g., people to positions). In
this case the number of rankings grows fast with the number of jobs
and resources.

The communication complexity issue, discussed above, motivates an
alternative analysis of the solution concepts for large VCG games.
We follow on two recent papers, Holzman et al \cite {RNDM} and
Holzman and Monderer \cite{HolMon2002}. We look for natural solution
concepts which are less demanding in terms of the agents
communication needs, yet are as almost convincing, in terms of the
incentive compatibility requirements, as weakly dominant strategies
(truth telling).
%ul
\subsection{Ex-Post Equilibrium}
\label{secexpost}

The solution concept of dominant strategies has the following
appealing properties:
\begin{itemize}
\item Agents act optimally no matter what other agents do.

\item The solution concept makes no use of any probabilistic
information on agents' valuations, either by the agents themselves
or by the mechanism designer.

\item Agent's strategies are robust to changes in the number of
players.

\end{itemize}

In addition, the solution is attained at truth telling strategies
which are quite simple (knowing $v_i$, agent $i$ does not need to do
any computation). Consequently we deduce that the chosen alternative
is the socially efficient alternative. We also note that by choosing
functions $h_i$ properly the game is individually rational and
budget balanced. An alternative, yet weaker, solution concept is
that of an ex-post equilibrium:

\begin{definition}
A tuple of strategies, $b$, is an {\em ex-post equilibrium}, for the
VCG game, $\Gamma(N,M,h,{\cal V})$, if for any player $i \in N$, any
valuation $v_i \in {\cal V}_i$ and any alternative strategy ${\hat
b}_i$ of $i$,
$$
U_i((b_i,b_{-i}),(v_i,v_{-i})) \ge U_i(({\hat
b}_i,b_{-i}),(v_i,v_{-i})) \ \ \ \forall v_{-i} \in {\cal
V}^{N-\{i\}}.
$$
\end{definition}

\begin{definition}
Fix a set of agents, $N$, a set of social alternatives $A$, and a
set of valuations, one for each agent, ${\cal V} \subset
{(\real^A)}^n$. A tuple of strategies, $b$, is an {\em ex-post
equilibrium for the class of VCG games, over $(N,A,{\cal V})$}, if
for all $N' \subset N$, $(b_j)_{j\in N'}$ is an ex-post equilibrium
for $\Gamma(N',M,h,{\cal V})$, for all $M \in {\cal M}_{N'}$ and $h
\in H$.
\end{definition}

Note that an ex-post equilibrium for the class of VCG games has
many of the properties of the solution concept of dominant
strategies:

\begin{itemize}
\item Agents act optimally no matter what other agents's
valuations are, as long as all agents keep to their strategies. In
other words, agents have no incentive to unilaterally deviate, even
after all valuations have been realized.

\item The solution concept makes no use of any probabilistic
information on agents' valuations.

\item Agent's strategies are robust to changes in the number of
players.
\end{itemize}

\begin{example}
%why 8? there are 9.
Consider a standard Vickrey auction auction of 2 goods, $A$ and $B$,
with 2 bidders. There are 9 possible allocations of the goods (one
can choose to allocate goods to none of the agents). Let ${\cal
V}_i$ be the set of $i$'s valuations that depend only on the goods
allocated to $i$ and are monotonically non-decreasing. Consider the
following strategies - Agent 1 announces the true valuation of the
grand bundle ($A$ and $B$), and the bundle $A$, and zero for $B$.
Agent $2$ announces his true valuation for the grand bundle and for
good $B$ and zero for $A$. These 2 strategies form an ex-post
equilibrium of the standard Vickrey auction. In fact, they form an
ex-post equilibrium for the class of VCG games, over $\cal
V$.\footnote{This example is taken from a comment in
\cite{HolMon2002} page 11.}
\end{example}

\begin{example}
Consider an auction of M goods and $N$ agents. As before, let
${\cal V}_i$ be the set of $i$'s valuations that depend only on
the goods allocated to $i$ and are monotonically non-decreasing.
Let $b_i$ be the strategy that assigns the true value for the
grand coalition and zero to all other allocations. Holzman et al
show that this is an ex-post equilibrium, for the class of VCG
games, over $\cal V$.
\end{example}

\section{Results}

Obviously any solution in weakly dominant strategies is also an ex
post equilibrium for the class of VCG games. However, Holzman et
al \cite{RNDM} have shown that there exist ex-post equilibria for
the class of VCG games induced by a combinatorial auction setting,
other than the dominant ones. In particular some of these
strategies have low communication complexity. A natural question
to pursue is whether this result is an artifact of the particular
structure of valuations of a combinatorial auction or whether it
can be extended to larger valuation sets.

In what follows we study the class of ex-post equilibria for
various classes of valuations. Our first result is immediate:

\begin{prp}\label{prop1} Assume $b$ is an ex-post equilibrium for the class of VCG games,
over (N,A,$\cal V$). If ${\cal V}'_i \subset {\cal V}_i$ then $b$
is an ex-post equilibrium for the class of VCG games, over
(N,A,${\cal V}')$.
\end{prp}

\begin{proof}
Follows directly from the definition.
\end{proof}

Unfortunately, if the set of valuation functions is large enough
then there are no ex post equilibria other than (near) truth
telling for any large class of VCG games:

%*************************************
\begin{thm}\label{NoPar1}
Assume that $n\ge 2$ and $|A|>2$, or alternatively that $n\ge 3$.
Then a strategy profile $b$ is an ex-post equilibrium for the
class of VCG mechanisms over $(N,A,{(\real^A)}^n)$ if and only if
$b_i(v_i)(a)=v_i(a)+f_i(v_i)$, for any arbitrary function $f_i:
\mathcal{V} \rightarrow \mathbb{R}$.
\end{thm}
%*************************************

In particular, note that the valuations reported by the agents may
differ from the true valuations, however, the difference between
the true valuation and the reported ones must be constant. We
shall refer to such strategies as {\it nearly truth telling over
A} (see definition \ref{def1} below).

The proof of theorem \ref{NoPar1} is composed of two distinct
proofs, one for the case $n\ge 2$ and $|A|>2$ and a different one
for the case $n\ge 3$.

One can hope that by adding more structure to the problem the
positive result can be salvaged. A valuation function is called {\it
non-negative} if $v_i(a) \ge 0$ for all $a \in A$. Let $\real_+^A$
be the set of all non-negative valuation functions. Indeed, this is
a natural property in the case of combinatorial auctions.
Nevertheless:

%**************************************************
\begin{thm}\label{NoPar2}
Assume that $n\ge 2$ and $|A|>2$, or alternatively that $n\ge 3$.
Then a strategy profile $b$ is an ex-post equilibrium for the
class of VCG games over $(N,A,{(\real_+^A)}^n)$ if and only if
$b_i(v_i)(a)=v_i(a)+f_i(v_i)$, for any arbitrary function
$f_i:\mathcal{V}\rightarrow \mathbb{R}$.
\end{thm}
%***************************************************

It is relatively straightforward to show that indeed the specific strategies prescribed
in theorems \ref{NoPar1} and \ref{NoPar2} are indeed an ex post equilibrium. The difficulty of the proof lies in the second direction.

\subsection{Constant Maximum Valuations}

We say that a set of valuations ${\cal V}_i$ {\em has a maximum} if
there exists some $a \in A$, called the {\em maximum}, such that
$v_i(a) \ge v_i(a')$ for all $a' \in A$ and all $v_i \in {\cal
V}_i$. Note that the set of non-decreasing valuations, in the
combinatorial auction setting, has a maximum. In particular, any
allocation that gives all the goods to agent $i$ is a maximum.

Valuations sets that have a maximum prevail in other settings as well. In the context of ordering a set of tasks, consider the valuations with the property that any agent want her own task to be processed first, but otherwise cares about which other tasks precede her own task. In the context of network construction, consider valuations where each player (who is a vertex in a graph) always prefers the star shaped graph centered around him over any other graph. Finally, in the context of facility location, assume agent always prefers all the `good facilities (e.g., library) to neighbor her and all the bad facilities (e.g., waste disposal) to be as far as possible.

Let ${\cal R}_i(a_i)$ be set of all non-negative valuation
functions for which $a_i$ is a maximum, and let ${\cal R}(\vec{a})
= \times_{i=1}^n {\cal R}_i(a_i)$.

The following family of strategies will play an important role for
this family of valuations.

\begin{definition}\label{def1}
Let $A' \subset A$ be a subset of social alternatives. A strategy
profile is called {\em nearly truth telling over $A'$} if:
\begin{itemize}
\item If $a \in A'$ then $b_i(v_i)(a) = v_i(a) + f(v_i)$ for some arbitrary function
$f_i:\cal V \to \real$;%
\footnote{Note that we do not exclude negative bids, namely $v_i(a) < - f(v_i)$.}
\item If $a \not \in A'$, then $b_i(v_i)(a) = C$, for an arbitrary 
$C\le \min_{a \in A'}b_i(v_i)(a)$.
\end{itemize}
\end{definition}

A nearly truth telling strategy prescribes telling the truth, up to
a shift in a constant, on some subset of pre-selected alternatives,
and assigns a valuation of zero to all other alternatives. So
players using the strategy need only communicate $|A'|+1$ numbers,
instead of $|A|$ numbers.

\begin{thm}\label{thm3}%
Consider the class of VCG games over $(N,A,{\cal R}(\vec{a}))$,
where $\vec{a}=(a_1,a_2,\ldots,a_n)$, and let $A' \subset A$ satisfy
$a_i \in A' \ \ \forall a_i$, $i=1,\ldots,,n$ (all the $n$ maxima
are in $A'$). Then any nearly truth telling strategy profile over
$A'$ is an ex-post equilibrium for this class.
\end{thm}

Unfortunately, not all ex post equilibria are nearly truth telling
for some subset $A'$. Consider the following strategy profile -
$b_i(v_i)(a_k) = v_i(a_k)+10 $ for all maxima, $a_k, \
k=1,\ldots,n$, and for all $a \not \in \{a_k: k=1,\ldots,n\}$,
$b_i(v_i)(a)$ is chosen arbitrarily in the interval $[0,9]$. We
leave it to the reader to verify that $b$ is an ex post equilibrium
for the class of VCG games over $(N,A,{\cal R}(a))$.

Although we are not able to characterize all ex-post equilibria for
the class of VCG games over $(N,A,{\cal R}(a))$, we can provide some
necessary conditions.

\begin{thm}\label{thm4}%
If $n\geq 3$ and the strategy profile, $b$, is an ex-post
equilibrium for the class of VCG games over $(N,A,{\cal R}(a))$ then
$b_i(v_i)(a_k) = v_i(a_k) + f(v_i)$, for all $i$ and $k$, where
$a_k$ is the maximum for player $k$.
\end{thm}

In words, in any ex-post equilibrium players (almost) report their
true valuations on the set of maxima.

\section{Efficiency}

It is quite obvious that even if ex post equilibria exist, as in the
models discussed in theorems \ref{thm5} and \ref{thm6}, the demand
on communication may be much smaller, compared with the dominant
strategy solution. In fact, agents may need as little as reporting
the value for $N$ alternatives only (compare $N={50}$ with $2^{50}$
or $50^{50}$ alternatives in example \ref{ex1}).

In section \ref{secexpost} we provided arguments why, conceptually,
the notion of ex post equilibrium is almost as robust as the
dominant strategy solution. However, when it comes to efficiency and
social welfare the two solution concepts differ. Whereas, the
dominant strategy solution maximizes social welfare (the sum of
agents´utilities) this is not so for many ex post equilibria.

\begin{example}
\label{ex5} Consider a complete information combinatorial auction setting with $N$
agents and $N$ goods. Assume agent $i$ values the bundle of goods, $K$,
as follows: $v_i(K)=0$ if $i \not \in K$, $v_i(K)= 1$ if $i \in K$
and $|K|< N$ and finally $v_i(K) = 1+\epsilon$ if $K$ is the grand bundle. Consider a strategy profile where each agent bids zero over any bundle that is not the grand bundle and truthfully on the grand bundle.
This is an ex post equilibrium of the combinatorial auction and the
communication complexity is very low. However this results
assigning the grand bundle randomly to one of the players, achieving
a social welfare of $1+\epsilon$, as opposed to the maximal social welfare that is achievable in the dominant strategy solution. Thus,
the efficiency ratio is $N$.
\end{example}

Let $S(m) = \sum_i v_i(m) $ denote the social welfare
for the social alternative $m$. Let $r(m,m') = \frac{S(m)}{S(m')}$.
For any VCG mechanism and any strategy profile $d$ we denote by
$VCG(d(v))$ the resulting social alternative, at the valuation
profile $v$. Recall that the dominant strategy profile $b$ maximizes $S$, namely
$S(VCG(b(v))) \ge S(m) \ \forall m$.
The following theorem
extends a result of Holzman et al \cite{RNDM}.

\begin{thm}
\label{thm5} Consider the class of VCG games over $(N,A,{\cal
R}(\vec{a}))$, where $\vec{a}=(a_1,a_2,\ldots,a_n)$, and let $A'
\subset A$ satisfy $a_i \in A' \ \ \forall a_i$, $i=1,\ldots,,n$
(all the $n$ maxima are in $A'$). Let $d$ be a nearly truth telling
strategy profile over $A'$ that is an ex-post equilibrium for this
class and let $b$ be the dominant strategy equilibrium. Then
$r(VCG(b(v)),VCG(d(v))) \le N$.
\end{thm}

\begin{proof}
The proof is similar to that of Remark $1$ in Holzman et al
\cite{RNDM}:
$$ s(VCG(b(v))) \le N \max_i v_i(VCG(b(v))) \le N \max_i(v_i(a_i)) \le N \max_i (\sum_j v_j(a_i))$$
where $a_i$ is the alternative $i$ prefers. On the other hand, $a_i$ is one of
the alternatives for which valuations are announced, therefore:
$$\max_i (\sum_j v_j(a_i))
 \le \max_i(v_i(VCG(d(v)))) \le
S(VCG(d(v)))$$ which completes the proof.
\end{proof}

Example \ref{ex5} shows that this bound is tight. In fact, we can
use the principles to that example to show that the efficiency loss
is not gradual and that one can get high efficiency loss even when
the communication complexity is very high:

\begin{example}
\label{ex6} Consider a setting with $N$ players and $M$ social
alternatives. Let $m_i$ denote the optimal social alternative for
$i$ and let $m_0 \not \in \{m_1,\ldots,m_N\}$ denote an arbitrary
alternative. Assume players play an ex post equilibrium with near
truth telling strategies on $M-\{m_0\}$. Now consider the following
valuation for player $i$ - $v_i(m_i)=1+i\epsilon$,$v_i(m_0)=1$ and
$v_i(m)=0$ for all other alternatives. For this vector of valuations
the resulting alternative is $m_N$ and the social welfare is
$1+N\epsilon$, whereas in he dominant strategy equilibrium the
resulting alternative is $m_0$ with a social welfare of $N$.
$r(VCG(b(v)),VCG(d(v)))$ approaches $N$ as $\epsilon$ aproaches zero
\end{example}

The bound we have shown is not a satisfactory one as the number of
players may be quite large. We consider two families of valuations
for which the efficiency loss is independent of the number of the
players.

The family of valuations $V=(V_1,\ldots,V_N)$ is called {\bf
homogeneous of degree $p$} if for any $v \in V$ and any $m \in M$
$\max_i(v_i(m)) < p\cdot \frac{\sum_iv_i(m)}{N}$. In words, for each
alternative there cannot be too much difference of opinion.%
\footnote{An alternative, and perhaps more appealing, definition
would involve the inequality $\max_i(v_i(m)) < p\cdot
(v_i(m)) \forall i$. However, we do not use this alternative as we do
not want to rule out the possibility for some agents to assign a
valuation of zero, while others assign a positive valuation.}
In many settings valuations are bounded, say $a \le v_i(m) \le b \ \ \forall i, v_i\in V_i, m \in A$. In such settings, if $a>0$ then valuations are homogeneous of degree $p=\frac{b}{a}$.
Another example for valuations of degree $p$, is in correlated settings where a common strictly positive signal is drawn and agents valuations are generated via idiosyncratic adjustments of the common signal.
More concretely think of a set of firms which compete for some public resource. The quality of the resource, and hence the potential revenues is common, yet the production costs, as a ratio of the revenues can be between $0<a$ and $b<1$. In this case homogeneity of degree $p=\frac{1-a}{1-b} $ prevails.

\begin{thm}
\label{thm6} Consider the class of VCG games over $(N,A,V)$, where
$V \subset {\cal R}(\vec{a})$ is homogeneous of degree $p$. Let $A'
\subset A$ satisfy $a_i \in A' \ \ \forall a_i$, $i=1,\ldots,,n$
(all the $n$ maxima are in $A'$). Let $d$ be a nearly truth telling
strategy profile over $A'$ that is an ex-post equilibrium for this
class and let $b$ be the dominant strategy equilibrium. Then
$r(VCG(b(v)),VCG(d(v))) \le p$.
\end{thm}

\begin{proof} Let $i_0$ denote the agent that values the alternative $VCG(b(v))$
the most and let $a_{i_0}$ be the alternative $i_0$ prefers.
$$ s(VCG(b(v))) \le N v_{i_0}(VCG(b(v))) \le N v_{i_0}(a_{i_0}))$$
By homogeneity $v_{i_0}(a_{i_0})) < p\cdot \frac{\sum_j
v_j(a_{i_0})}{N}$ which implies that
$$s(VCG(b(v))) < p\cdot \sum_j v_j(a_{i_0})$$.

Because $a_{i_0}$ is one of the alternatives for which valuations
are announced it holds that:
$$\sum_j v_j(a_{i_0})) \le  \sum_j v_j(VCG(d(v)))$$
which completes the proof.
\end{proof}

Another family of valuations which we study is one where players´ valuations differ
significantly over each alternative. A family of valuations
$V=(V_1,\ldots,V_N)$ is called {\bf compatible of degree $p$} if for
any $v \in V$ and any $m \in M$ there are at most $p$ players for
which $v_i(m) > 0$. As an example consider a combinatorial auction with $p$ goods.

\begin{thm}
Consider the class of VCG games over $(N,A,V)$, where $V \subset
{\cal R}(\vec{a})$ is compatible of degree $p$. Let $A' \subset A$
satisfy $a_i \in A' \ \ \forall a_i$, $i=1,\ldots,,n$ (all the $n$
maxima are in $A'$). Let $d$ be a nearly truth telling strategy
profile over $A'$ that is an ex-post equilibrium for this class and
let $b$ be the dominant strategy equilibrium. Then
$r(VCG(b(v)),VCG(d(v))) \le p$.
\end{thm}

The proof of this theorem mimics the proof of Theorem \ref{thm5},
with $p$ replacing $N$, and is therefore omitted. Note that in
combinatorial auctions the number of players that have a positive
valuation for any alternative is at most the number of goods.
Additionally, in any bundling equilibrium derived from a partition
of the set of goods, the number of players that have a positive
valuation for any alternative is at most the size of the partition.

\section{Combinatorial Auctions}

An analysis of ex post equilibria in VCG mechanisms for the setting of combinatorial auctions is provided in In Holzman, Kfir-Dahav, Monderer and Tennenholtz \cite{RNDM} and  Holzman and Monderer \cite{HolMon2002}. These papers focus on combinatorial auctions with 3 bidders or more with monotonic valuations. Their main finding is that the ex post equilibria of such auctions are characterized by submitting bids on a subset of the possible bundles (this is referred to as a bundling equilibrium), which is a quasi-field. Namely, it is non-empty set of sets that is closed under complements and under disjoint unions.

Note that a social alternative in the auction setting is an assignment of the set of all goods to the set of players (the bidders and the seller). However, a player's valuation depends only on the goods allocated to her (there are no externalities). Therefore specifying agent $i$'s valuation for a bundle $B$ induces valuations for all social alternatives in which agent $i$ receives the bundle $B$. In addition, monotonic valuations imply that allocating the grand bundle to agent $i$ always maximizes $i$'s valuation over the possible social alternatives and so a set of maximizers is identified.

This unique structure allows for a full characterization of the ex-post equilibria, in contrast with our partial characterization for the general case. Comparing the results for the general case with those of the auction setting is not obvious. To see this consider the following example:

\begin{example}
Consider an auction with 3 goods, $\{a,b,c\}$ and 3 players. Consider the subset of social alternatives:
$$S' = \{(\emptyset,\emptyset,\emptyset),(\emptyset,bc,\emptyset), (\emptyset,abc,\emptyset), (ab,\emptyset,\emptyset), (abc,\emptyset,\emptyset), (\emptyset,\emptyset,abc) \}.$$
The set $S'$ includes all the three maximizers,$(abc,\emptyset,\emptyset), (\emptyset,abc,\emptyset)$ and $(\emptyset,\emptyset,abc)$. Therefore, by Theorem \ref{thm3} this set induces an ex post equilibrium, where players bid truthfully over this set. On the other hand, players do not submit bids on a quasi-field. In fact, note that this ex post equilibrium is not a bundling equilibrium as players do not bid on the same bundles. This seems to contradict the findings of Holzman, Kfir-Dahav, Monderer and Tennenholtz \cite{RNDM} and  Holzman and Monderer \cite{HolMon2002}.
\end{example}

Is this a real contradiction? The answer is clearly no. To settle this note that when we cast our general model to the combinatorial auction setting we do not assume additional restrictions on valuation functions. In particular agents valuation may have externalities and need not be monotonic. Therefore, valuations and bids over $S'$ are silent about valuations outside of $S'$.

If, however, we adopt the two restrictions of monotonicity and no-externality then the valuations of $S'$ extend to additional social alternatives. For example, if player 1 bids $v$  on the social alternative $(ab,\emptyset, \emptyset)$ then this implies a bid of $v$ on the social alternative $(ab,c,\emptyset) \not \in S'$. However, the valuation of player 2 for this social alternative cannot be deduced from her valuations over $S'$, making the bids asymmetric. This, in turn, makes Theorem \ref{thm3} mute as the conditions do not hold.

\section{Appendix - Proofs}

We begin by some preliminary observations needed to prove our main results.

\subsection{Preliminary Observations}

We observe that  in any ex post equilibrium the most valued
alternative for a player must have the highest reported
valuation.

\begin{lemma}\label{l1}
Let $n \ge 1$ and let $b$ be an ex-post equilibrium for the class of
VCG games over $(N,A,{\cal V})$. Then for all $i$ and all $v_i \in
{\cal V}_i$, if $v_i(a) > v_i(a')$, for all $a' \not = a$ then
$b_i(v_i)(a) > b_i(v_i)(a')$ for all $a' \not =a$.
\end{lemma}

\begin{proof} Assume $v_i(a) > v_i(a')$, for all $a' \not = a$, and
consider the one player game with player $i$. In this game, the
chosen alternative must be optimal for $i$, and so it must be that
$i$'s valuation on it was the highest, namely $b_i(v_i)(a) >
b_i(v_i)(a')$ for all $a' \not =a$.
\end{proof}

\begin{lemma}\label{old5} Let $b$ be an ex-post
equilibrium for the class of VCG games over $(N,A,{\cal V})$. If $a$
is chosen at the profile $v$, then for any $i$,
$$v_i(a) +
\sum_{j\not = i}b_j(v_j)(a) \  \ge \ v_i(a') + \sum_{j\not =
i}b_j(v_j)(a'), \ \  \forall a' \in A.$$
\end{lemma}

\begin{proof} Assume the claim is wrong, and that for some $i$
$v_i(a) + \sum_{j\not = i}b_j(v_j)(a) < v_i(a') + \sum_{j\not =
i}b_j(v_j)(a') \le \max_{{\hat a} \in A} v_i({\hat a}) + \sum_{j\not
= i}b_j(v_j)({\hat a}) $. Note that the left hand side is $i$'s
utility, whereas the right hand side is $i$'s utility from reporting
truthfully. Thus, contradicting the ex-post equilibrium assumption.
\end{proof}

\begin{lemma}\label{old2}
Let $n \ge 2$ and let $b$ be an ex-post equilibrium for the class of
VCG games over $(N,A,{(\real_+^A)}^n)$. Then for all $i$ and all
$v_i \in \real_+^A$, if $v_i(a)= v_i(a')$, then $b_i(v_i)(a) = b_i(v_i)(a')$.%
\end{lemma}

\begin{proof}: Assume $v_i(a) = v_i(a')$ but $b_i(v_i)(a) >
b_i(v_i)(a')$. Let $\bar x = \max_{{\hat a}\in A} |b_i(v_i)({\hat a})|$.
Consider the following valuation function for some player $j \not
=i$. $v_j({\bar a})= 0$, for all ${\bar a} \not \in \{a,a'\}$, $\
v_j(a) = 3\bar x$ and $v_j(a') = 3\bar x +{b_i(v_i)(a) - b_i(v_i)(a') \over
2}$. Note that $v_j(a') > v_j(a)$.

Assume that some ${\hat a} \not \in \{a,a'\}$ is chosen in the two
player game with $i$ and $j$. Then player $j$'s utility does not
exceed $\bar x$. However, by bidding truthfully either $a$ or $a'$ would
have been chosen and $j$'s utility would be at least $2\bar x$, leading
to a contradiction. Therefore, either $a$ or $a'$ are chosen.

If $a'$ is chosen then $j$'s utility is $3\bar x +{b_i(v_i)(a) -
b_i(v_i)(a') \over 2} + b_i(v_i)(a') $. This is strictly less than
$3\bar x + b_i(v_i)(a)$, which is the utility $j$ could have received by
reporting truthfully on $a$ and zero on all other alternatives. This
contradicts the ex-post equilibrium assumption, and therefore it
must be the case that $a$ is chosen. In this case the utility of $i$
is $v_i(a) + b_j(v_j)(a)$. By our assumption $v_i(a) = v_i(a')$. By
lemma \ref{l1}, $b_j(v_j)(a) < b_j(v_j)(a')$. Therefore the utility
of $i$ is strictly less than $v_i(a') + b_j(v_j)(a')$, which is what
$i$ could have received by reporting truthfully on
$a'$ and zero on all other alternatives.%
\end{proof}

\subsubsection{The mean value exclusion property and parallelograms}

The main lemma we will use is a simple observation due to Monderer
and Holzman \cite{HolMon2002}, which they refer to as the ``mean
value exclusion" property.

\begin{definition}
The pair of functions $f_1,f_2: \mathbb{R_+} \to \mathbb{R_+}$
satisfies the \emph{mean value exclusion} condition if $\forall{s,t,y}\geq{0}$
\[
s<y\leq{f_1(s)} \; \mbox{or} \; f_1(s)\leq{y}<s\Rightarrow
y\neq{f_2(t)}.
\]
and symmetrically
\[
s<y\leq{f_2(s)} \; \mbox{or} \; f_2(s)\leq{y}<s\Rightarrow
y\neq{f_1(t)}.
\]
\end{definition}

Let $I \subset \mathbb{R}$ denote an open interval. We denote its
closure by ${\bar I}$, its supremum by $I^+ = \sup_{x\in I}x$ and
its infimum, by $I_- = \inf_{x \in I}x$.

Let $\Omega$ denote the union of disjoint open intervals in
$\mathbb{R}_+$ such that for any $I \in \Omega$ $I^- \not = 0$.
Consider an arbitrary function $G:\Omega \to \{-1,+1\}$ satisfying
$G(I_1)\times G(I_2) = -1$ whenever $I_1^+ = I_2^-$. We say that a
pair of function $h_i:\mathbb{R}_+\rightarrow \mathbb{R}_+$,
$i=1,2$, is \emph{$(\Omega,G)$-compatible} if it satisfies the
following conditions:
\begin{enumerate}
\item If $x \in \mathbb{R}_+ \backslash \displaystyle{\bigcup_{I \in{\Omega}}}{\bar{I}}$,\;
then $h_1(x) = h_2(x)= x$
\item If $x\in I \in \Omega$ and $G(I) = -1$ then $h_1(x)=I^-$ and $h_2(x)=I^+$
\item If $x\in I \in \Omega$ and $G(I) = +1$ then $h_1(x)=I^+$ and
$h_2(x)=I^-$
\item If $I \in {\Omega}$ and $I^- \not = J^+$ for all $J \in \Omega$,
then:
\begin{enumerate}
\item if $G(I)= -1$ then $h_1(I^-) = I^-$ and $h_2(I^-) = I^+$.
\item if $G(I)= +1$ then $h_1(I^-) = I^+$ and $h_2(I^-) = I^-$.
\end{enumerate}
\item If $I \in {\Omega}$ and $I^+ \not = J^-$ for all $J \in \Omega$,
then:
\begin{enumerate}
\item if $G(I)= -1$ then $h_1(I^+) = I^-$ and $h_2(I^+) = I^+$.
\item if $G(I)= +1$ then $h_1(I^+) = I^+$ and $h_2(I^+) = I^-$.
\end{enumerate}

\item If $I^+ =J^-$ for some $I \not = J \in \Omega$ and $G(I)= +1$ then
$h_1(I^+)=I^+$ and $h_2(I^+)\in \{I^-,J^+\}$.

\item If $I^+ =J^-$ for some $I \not = J \in \Omega$ and $G(I)= -1$ then
$h_1(I^+)\in \{I^-,J^+\} $ and $h_2(I^+)=I^+$.

\end{enumerate}

Note that a set of disjoint open intervals, $\Omega$, and a
function $G:\Omega \to \{-1,+1\}$ almost determines the pair of
functions which is $(\Omega,G)$-compatible. The different variants
of such a pair of functions stems from conditions (6) and (7)
which allow a choice between two possible values.

\begin{prp}[Parallelogram]\label{Pprp}
Suppose that $g_i:\mathbb{R}_+\rightarrow \mathbb{R}_+, i=1,2,$
satisfy the mean value exclusion condition. Then there exists a set
of disjoint open segments $\Omega$ and a function $G:\Omega \to
\{-1,+1\}$, with the restriction that  $I^+ = J^-$ implies $G(I)
\times G(J) = -1$, such that the pair $(g_1,g_2)$ is
$(\Omega,G)$-compatible function.
\end{prp}

To construct this family of segments from the given $g_1,g_2$ we
use the following definitions and lemmas:
\begin{defn}
$D_1$: A set of open segments such that\\
$\forall I \in D_1, g_1(I^-)=I^+$.
\end{defn}
\begin{defn}
$D_2$: A set of open segments such that\\
$\forall I \in D_2, g_1(I^+)=I^-$.
\end{defn}
\begin{defn}
$D_3$: A set of open segments such that\\
$\forall I \in D_3, g_1(I^-)=I^+$ and there is no $x'<I^-$ such that
$g_1(x')=I^+$.
\end{defn}
\begin{defn}
$D_4$: A set of open segments such that\\
$\forall I \in D_4, g_1(I^+)=I^-$ and there is no $y'>I^+$ such that
$g_1(y')=I^-$.
\end{defn}
\begin{defn}
$D_5$: A set of open segments such that $\forall I \in D_5$ there
exists a monotone decreasing sequence $\{x_k\}_{k=1}^\infty,
x_k\rightarrow I^-$ and $\forall{k} \in\mathbb{N}\; g_1(x_k)=I^+$,
but $g_1(x')\neq I^+$ for all $x'\leq I^-$.
\end{defn}
\begin{defn}
$D_6$: A set of open segments such that $\forall I \in D_5$ there
exists a monotone increasing sequence $\{y_k\}_{k=1}^\infty,
y_k\rightarrow I^+$ and $\forall{k} \in\mathbb{N}\; g_1(y_k)=I^-$,
but $g_1(y')\neq I^-$ for all $y'\geq I^+$.
\end{defn}
\begin{defn}\label{defD}
$ D= D_3 \cup D_4 \cup D_5 \cup D_6$.
\end{defn}
\begin{defn}
We say that a segment $I$ satisfies the \emph{$"+"$ condition} if
for all $t \in I, g_1(t)=I^+, g_2(t)=I^-$.
\end{defn}
\begin{defn}
We say that a segment $I$ satisfies the \emph{$"-"$ condition} if
for all $t \in I, g_1(t)=I^-, g_2(t)=I^+$.
\end{defn}

When a family of segments will be created in the sequel the "+" and
"-" conditions will be attached to a segment with a function G,
which will give a segment a +1 if it satisfies the "+" condition and
a -1 if it satisfies the "-" condition.

\begin{lemma}\label{D}
$\forall{I}\in{D_1 \cup D_2} \; \exists I' \in{D} $ such that $I
\subseteq I'.$
\end{lemma}

\begin{proof}[Proof of Lemma \ref{D}] Let $I\in D_1\cup D_2$. We first
assume that $g_1(I^-)=I^+$, and let $x_0=\inf \{x'|g_1(x')=I^+\}$,
of course $x_0\leq I^-$. If $g_1(x_0)=I^+$ then $I \subseteq
(x_0,I^+)\in D_3$. Else there exists a sequence
$\{x_k\}_{k=1}^\infty $ that converges to $x_0$, such that $\forall
k>0 \;g_1(x_k)=I^+, \forall x'\leq x_0\; g_1(x')\neq I^+$. So we
have $I\subseteq (x_0,I^+)\in D_5$. Now we assume that
$g_1(I^+)=I^-,$ and let $y_0=\sup \{y'|g_1(y')=I^-\}$, of course
$y_0\geq I^+$. If $g_1(y_0)=I^-$ then $I \subseteq (I^-,y_0)\in
D_4$. Else there exists a sequence $\{y_k\}_{k=1}^\infty $ that by
Lemma \ref{lemy} converges to $y_0<\infty$, such that $\forall k>0\;
g_1(y_k)=I^-, \forall y' \geq y_0 \; g_1(y')\neq I^-$. So, we have
$I\subseteq (I^-,y_0)\in D_6$.
\end{proof}

\begin{lemy}\label{lemy}
Let $\{y_k\}_{k=1}^\infty$ be a sequence such that $\forall{k}\in
\mathbb{N}\; x<y_k$ and $g_1(y_k)=x$ and $y_k\rightarrow y$.
Then $y<\infty$.\\
\end{lemy}

\begin{proof}[Proof of Lemma \ref{lemy}]Assume for the sake of
contradiction that $y_k\rightarrow \infty$ and choose $t$ such
that $t>x$.
Then we shall see where $g_2(t)$ can be:\\
$g_2(t)\notin [x,\infty)$, otherwise we can find a $y_m$ such that
$x\leq g_2(t)<y_m, g_1(y_m)=x \Rightarrow g_1(y_m)\leq
g_2(t)<y_m$,
a contradiction to mean value exclusion.\\
$g_2(t)\notin [0,x)$, otherwise $g_2(t)<x<t, g_1(y_0)=x
\Rightarrow g_2(t)\leq g_1(y_0)<t$, a contradiction.\\
\end{proof}

\begin{lemma}\label{idt}
Let $x \in \mathbb{R}_+ \backslash \displaystyle{\bigcup_{I
\in{D}}}{\bar{I}}$ where $\bar{I}$ is the closure of $I$ . Then
$g_1(x)=g_2(x)=x$.\\
\end{lemma}

\begin{proof}[Proof of Lemma \ref{idt}] Let $x\in \mathbb{R}_+ \backslash
\displaystyle{\bigcup_{I\in{D}}}{\bar{I}}$. Then $g_1(x)=x$,
otherwise, if $x<g_1(x)$ then $(x,g_1(x))\in D_1$ and by Lemma
\ref{D} it follows that there exists $I\in D$ such that
$(x,g_1(x))\subseteq I $, and hence $x\in \bar{I}$, contradicting
the assumption. The same goes for $x>g_1(x)$.\\
If $x<g_2(x)$ then let $x<t<g_2(x)$. We shall see where $g_1(t)$ can be:\\
$g_1(t)\notin [g_2(x),\infty)$, otherwise $t<g_2(x)\leq
g_1(t)$, a contradiction.\\
$g_1(t)\notin (x,g_2(x))$, otherwise $x<g_1(t)<
g_2(x)$, a contradiction.\\
It follows that $g_1(t)\leq x<t$ and hence $(g_1(t),t)\in D_2$. By
Lemma \ref{D} $\exists I\in D$ such that $(g_1(t),t)\subseteq I$
and therefore $ x\in \overline{I}$,
contradicting the assumption.\\
If $g_2(x)<x$ then let $g_2(x)<t<x$. We shall see where $g_1(t)$ can be:\\
$g_1(t)\notin [0,g_2(x)]$, otherwise $g_1(t)\leq g_2(x)<t$, a contradiction.\\
$g_1(t)\notin [g_2(x),x)$, otherwise $g_2(x)\leq g_1(t)<x$, a contradiction.\\
It follows that $t<x\leq g_1(t)$ and hence $(t,g_1(t))\in D_1$. By
Lemma \ref{D} $\exists I\in D$ such that $(t,g_1(t))\subseteq I$
and therefore $x\in \overline{I}$,
contradicting the assumption.\\
\end{proof}

\begin{lemma}\label{plmi}
Let $I\in D$. Then $I$ satisfies the "+" condition or the "-"
condition.\\
If $I$ satisfies the "+" condition then $g_2(I^-)=I^-$ and
$g_1(I^+)=I^+$.\\
If $I$ satisfies the "-" condition then $g_1(I^-)=I^-$ and
$g_2(I^+)=I^+$.
\end{lemma}

\begin{proof}[Proof of Lemma \ref{plmi}] Let $I\in D$. We will split
the proof into two parts:\\
Part 1: if $I\in D_3\cup D_5$ then $I$ satisfies the "+" condition
and $g_2(I^-)=I^-$, $g_1(I^+)=I^+$.\\
Part 2: if $I\in D_4\cup D_6$ then $I$ satisfies the "-"
condition and $g_1(I^-)=I^-$, $g_2(I^+)=I^+$.\\

Proof of part 1:
\begin{enumerate}
\item $\forall t, I^-\leq t<I^+$, we have $g_2(t)\leq I^-$:
Suppose this is not true. If $I\in D_3$ then $g_1(I^-)=I^+$ so if
$I^-<g_2(t)\leq g_1(I^-)$ it will be a contradiction. If
$t<g_1(I^-)<g_2(t)$ it will also be a contradiction. If $I\in D_5$
and $I^-<g_2(t)\leq I^+$ then we shall look at $x_k$ of the sequence
(that is given with a $I\in D_5$) such that $I^- < x_k<g_2(t)$ and
$g_1(x_k)=I^+$. Then $x_k<g_2(t)\leq g_1(x_k)$, a contradiction. If
$t<I^+<g_2(t)$ then again we shall look at the same $x_k$ and we
will get that $t<g_1(x_k)<g_2(t)$, a contradiction.\\

\item $\forall t, I^-\leq t<I^+$, we have $g_2(t)=I^-$: Indeed, suppose there
exists $I^-\leq t<I^+$ such that $g_2(t)<I^-$. Choose $s$ such that
$g_2(t)<s<I^-\leq t$. We shall see where $g_1(s)$ can be:\\
$g_1(s)\notin [0,g_2(t)]$: otherwise $g_1(s)\leq g_2(t)<s$, a
contradiction.\\
$g_1(s)\notin (g_2(t),t)$: otherwise $g_2(t)<g_1(s)<t$, a
contradiction.\\
$g_1(s)\notin [t,I^+)$: otherwise we will find $t'$ such that
$I^-\leq g_1(s)<t'<I^+$, then from $(1.)$ it follows that
$g_2(t')\leq I^-$, so $g_2(t')\leq I^-\leq g_1(s)<t'$, a contradiction.\\
$g_1(s)\notin (I^+,\infty)$: otherwise choose $t', I^+<t'<g_1(s)$.
We shall see where $g_2(t')$ can be:\\
If $g_2(t')\leq I^+$ we  have that
\[
        g_2(t')\leq I^+=\left\{
                        \begin{array}{ll}
                            g_1(I^-)<t'       & \mbox{$I\in D_3$}\\
                            g_1(x_k)<t'     & \mbox{$I\in D_5$}
                        \end{array}
                      \right.
\]
a contradiction.\\
If $I^+<g_2(t')\leq g_1(s)$ we  have that $s<g_2(t')\leq g_1(s)$, a
contradiction.\\
If $g_1(s)<g_2(t')$ we  have that
$t'<g_1(s)<g_2(t')$, a contradiction.\\

The only remaining possibility is $g_1(s)=I^+$. But since $s<I^-$
and $I\in D_3\cup D_5$, this is impossible.

\item $\forall t, I^-<t\leq I^+$, we have $g_1(t)\geq I^+$: Otherwise there
exists $x_0, I^-<x_0\leq I^+$ such that $I^-\leq g_1(x_0)<I^+$ or
$g_1(x_0)<I^-$.\\
If $g_1(x_0)<I_-$ then let $t_0$ be a number that satisfies
$(2.)$. We have that\\
$g_1(x_0)<I^-=g_2(t_0)<x_0$, a contradiction.\\
If $I^-\leq g_1(x_0)<I^+$ then choose $t_0, g_1(x_0)<t_0<I^+$. It
follows from $(2.)$ that \\
$g_2(t_0) = I^-\leq g_1(x_0)<t_0$, a contradiction.\\

\item $\forall t, I^-<t\leq I^+$, we have $g_1(t)=I^+$: Otherwise there
exists $t_0, I^-<t_0\leq I^+$ such that $g_1(t_0)>I^+$. Choose $s,
I^+<s<g_1(t_0)$. We shall see where $g_2(s)$ can be:\\
$g_2(s)\notin [g_1(t_0),\infty)$: otherwise $s<g_1(t_0)\leq
g_2(s)$, a contradiction.\\
$g_2(s)\notin (t_0,g_1(t_0))$: otherwise $t_0<g_2(s)<g_1(t_0)$, a
contradiction.\\
$g_2(s)\notin [0,t_0]$: otherwise
\[
        g_2(s)\leq t_0 \leq I^+=\left\{
                        \begin{array}{ll}
                            g_1(I^-)<s       & \mbox{$I\in D_3$}\\
                            g_1(x_k)<s     & \mbox{$I\in D_5$}
                        \end{array}
                      \right.
\]
a contradiction.\\
\end{enumerate}

    Proof of part 2:\\
As the mean value exclusion condition is symmetric, by reversing the
order of $\mathbb{R}_+$ and exchanging $I^-$ and $I^+$ the proof of
part 1 yields a proof of part 2.
\end{proof}

\begin{lemma}\label{disj}
Let $I_1,I_2\in D$ such that $ I_1\neq I_2$. Then $I_1\cap I_2=\phi$.\\
\end{lemma}

\begin{proof}[Proof of Lemma \ref{disj}]Let us assume for the sake of
contradiction that $t\in I_1\cap I_2 \neq \phi$. Then
if $I_1=(I_1^-,I_1^+), I_2=(I_2^-,I_2^+)$ we have three cases:\\
\begin{enumerate}
\item $I_1,I_2$ both satisfy the "+" condition. Then:\\
$g_1(t)=I_1^+$ and $g_1(t)=I_2^+$\\
$g_2(t)=I_1^-$ and $g_2(t)=I_2^-$\\
$\Rightarrow I_1^+=I_2^+, I_1^-=I_2^- \Rightarrow I_1=I_2$, a
contradiction.

\item $I_1,I_2$ both satisfy the "-" condition. Then:\\
$g_1(t)=I_1^-$ and $g_1(t)=I_2^-$\\
$g_2(t)=I_1^+$ and $g_2(t)=I_2^+$\\
$\Rightarrow I_1^+=I_2^+, I_1^-=I_2^- \Rightarrow I_1=I_2$, a
contradiction.
\item $I_1$ satisfies the "-" condition, and $I_2$ satisfies the
"+" condition then:\\
$I_1$ satisfies the "-" condition $\Rightarrow g_1(t)=I_1^-$ and $g_2(t)=I_1^+$\\
$I_2$ satisfies the "+" condition $\Rightarrow g_1(t)=I_2^+$ and $g_2(t)=I_2^-$\\
$\Rightarrow I_1^-=I_2^+, I_1^+=I_2^-$. Hence one of the segments is
not defined as a legal segment, a
contradiction.\\

The symmetric case to (3.) has a symmetric proof.
\end{enumerate}
\end{proof}

\begin{lemma}\label{endp}
Let t be an end point of a segment $I\in D$. Then:
\begin{enumerate}

\item If t is an end point of I alone then $\forall x\in I, g_1(t)=g_1(x),g_2(t)=g_2(x)$.
\item If t is an end point of two segments $I,J\in D$ then:
\begin{enumerate}
\item The two segments have opposite signs.
\item $\forall x\in I, g_1(t)=g_1(x)$ or $\forall x\in J, g_1(t)=g_1(x)$
and $\forall x\in I, g_2(t)=g_2(x)$ or $\forall x\in J,
g_2(t)=g_2(x)$.
\end{enumerate}
\end{enumerate}
\end{lemma}

\begin{proof}[Proof of Lemma \ref{endp}]
If $t$ is an end point of $I$ alone then we shall split the proof to
4 parts:
\begin{enumerate}
\item $I$ satisfies the "+" condition and $t=I^-$. Then by
Lemma \ref{plmi} it follows that
$g_2(I^-)=I^-$, we will show that $g_1(I^-)=I^+$:\\
$g_1(I^-)\notin [I^-,I^+)$: Otherwise, we will find a number $s,
g_1(I^-)<s<I^+$, and then by Lemma \ref{plmi} it follows that
$g_2(s)=I^-$. This implies $g_2(s)=I^-\leq g_1(I^-)<s$,
a contradiction.\\
$g_1(I^-)\notin [0,I^-)$: Otherwise $(g_1(I^-),I^-)\in D_2$ and by
Lemma \ref{D} it follows that there exists a segment $I'\in D$ such
that $(g_1(I^-),I^-)\subseteq I'$. But by Lemma \ref{disj} $I\cap
I'=\phi$. Hence  $I^-$ is the right end point of $I'$, a
contradiction to the assumption of this case.\\
$g_1(I^-)\notin (I^+,\infty)$: Otherwise $(I^-,g_1(I^-))\in D_1$ and
therefore $ \exists I'\in D$ such that $I \varsubsetneq
(I^-,g_1(I^-))\subseteq I'$, a contradiction to Lemma \ref{disj}.
The only remaining possibility is $g_1(I^-)=I^+$.
\item
$I$ satisfies the "+" condition and $t=I^+$. Then by Lemma
\ref{plmi}
it follows that $g_1(I^+)=I^+$, we will show that $g_2(I^+)=I^-$:\\
$g_2(I^+)\notin (I^-,I^+]$: Otherwise, we will find a number $s,
I^-<s<g_2(I^+)$, and then by Lemma \ref{plmi} it follows that
$g_1(s)=I^+$. This implies $s<g_2(I^+)\leq I^+=g_1(s)$, a contradiction.\\
$g_2(I^+)\notin [0,I^-)$: Otherwise, the segment $(g_2(I^+),I^+)$ is
not contained in $I$. We will show that there exists another $J\in
D$ such that $(g_2(I^+),I^+)\subseteq J$. This will be a
contradiction to Lemma \ref{disj}. To  show the existence of such a
segment we shall show that for $s$ such that $g_2(I^+)<s<I^-$,
$g_1(s)\geq I^+$. This will imply by using Lemma \ref{D} that there
exists a segment $J\in D$ as desired. Let $s$ satisfy
$g_2(I^+)<s<I^-$. We will show that all other possibilities cannot
be true:\begin {tabbing} aaa \= \kill
    \>$g_1(s)\notin [0,g_2(I^+)]$: Otherwise $g_1(s)\leq g_2(I^+)< s$, a contradiction.\\
    \>$g_1(s)\notin (g_2(I^+),I^+)$: Otherwise $g_2(I^+)<g_1(s)<I^+$, a contradiction.\\

$g_2(I^+)\notin (I^+,\infty)$: Otherwise,  we will show that there\\
exists a second segment $J\in D$ such that $J\neq I$ but $I^+$
is a left\\
end point of $J$. Let $s\in (I^+,g_2(I^+))$. Then\\
aaa \=                    \kill
    \>$g_1(s)\notin [0,I^-]$: Otherwise $g_1(s)\leq I^-=g_2(I^-)<s$, a contradiction.\\
    \>$g_1(s)\notin I$: Otherwise, we can find a number $d\in (g_1(s),I^+)$\\
    \>and then by Lemma \ref{plmi} $g_2(d)=I^-< g_1(s)<d$, a contradiction.\\
    \>$g_1(s)\notin (I^+,g_2(I^+)]$: Otherwise $I^+<g_1(s)\leq g_2(I^+)$, a contradiction.\\
    \>$g_1(s)\notin (g_2(I^+),\infty)$: Otherwise $s<g_2(I^+)<g_1(s)$, a contradiction.\\
So, by Lemma \ref{D} there exists a segment $J\in D$ such that
$(I^+,g_2(I^+))\subseteq J$.\\
By Lemma \ref{disj} $J\bigcap I=\phi $, so $I^+$ is the left end
point of $J$ and $I$,\\
a contradiction to the assumption of this case.\\
The only remaining possibility is $g_2(I^+)=I^-$.
\end{tabbing}
\item $I$ satisfies the "-" condition and $t=I^-$: the proof
is similar to (2.).
\item $I$ satisfies the "-" condition and $t=I^+$: the proof
is similar to (1.).
\end{enumerate}
    If $t$ is an end point of two segments, we split the proof to two
parts:
\begin{enumerate}
\item The two segments have opposite signs:\\
Assume for the sake of contradiction that $t$ is an end point of two
segments $I_1=(x,t), I_2=(t,y)$ that both satisfy the "+" condition.
Then, by Lemma \ref{plmi} and the fact that $t$ is the left end
point of $I_2$, it follows that $g_2(t)=t$. It also follows by Lemma
\ref{plmi} that $\forall x' \in (x,t)\;
g_1(x')=t$. Hence $x'<t=g_2(t)=g_1(x')$, a contradiction.\\
In the same way it can be shown that $t$ can't be an end point of
2 segments that satisfy the "-" condition.
\item \begin{tabbing} Now we shall show that\\
$\forall x\in I_1, g_1(t)=g_1(x)$ or $\forall x\in I_2, g_1(t)=g_1(x)$ and\\
$\forall x\in I_1, g_2(t)=g_2(x)$ or $\forall x\in I_2, g_2(t)=g_2(x)$:\\
Let us say that $t$ is a common end point of $I_1=(x,t)$ that
satisfies the "+"\\
condition, and of $I_2=(t,y)$ that satisfies the "-"
condition.\\
(The opposite case is handled in a similar way.)\\
By Lemma \ref{plmi} and the fact that $I_1$ satisfies the
"+" condition,\\
it follows that $g_1(t)=t$ as for any $x'\in (x,t)$.\\
So we need to show that $g_2(t)\in \{x,y\}$:\\
$g_2(t)\notin [0,x)$: Otherwise, choose $s, g_2(t)<s<x$. We shall
see where $g_1(s)$ can be:\\
aaa \=                    \kill
    \>$g_1(s)\notin [0,g_2(t)]$: Otherwise $g_1(s)\leq g_2(t)<s$, a contradiction.\\
    \>$g_1(s)\notin (g_2(t),t)$: Otherwise $g_2(t)<g_1(s)<t$, a contradiction.\\
    \>$g_1(s)\notin [t,\infty)$: Otherwise choose $y', x<y'<t$. \\
By Lemma \ref{plmi} it follows that\\
$g_2(y')=x$ and hence $s<g_2(y')<t\leq g_1(s)$, a contradiction.\\
$g_2(t)\notin (x,t]$: Otherwise choose $y', x<y'<g_2(t)$.\\
By Lemma \ref{plmi} it follows that\\
$g_1(y')=t$ and hence $y'<g_2(t)\leq t=g_1(y')$, a contradiction.\\
$g_2(t)\notin (t,y)$: Otherwise choose $y', g_2(t)<y'<y$.\\
By Lemma \ref{plmi} it follows that\\
$g_1(y')=t$ and hence $g_1(y')=t<g_2(t)<y'$, a contradiction.\\
$g_2(t)\notin (y,\infty)$: Otherwise choose $s, y<s<g_2(t)$.\\
We shall see where $g_1(s)$ can be:\\
    aaa \=                    \kill
    \>$g_1(s)\notin [0,y]$: Otherwise choose $y', t<y'<y$.\\
By Lemma \ref{plmi} it follows that\\
    \>$g_2(y')=y$ and hence $g_1(s)\leq g_2(y')<s$, a contradiction.\\
    \>$g_1(s)\notin (y,g_2(t)]$: Otherwise $t<y<g_1(s)\leq g_2(t)$,\\
    a contradiction.\\
    \>$g_1(s)\notin [g_2(t),\infty)$: Otherwise $s<g_2(t)\leq g_1(s)$,\\
    a contradiction.\\
We have shown that  $g_2(t)\in \{x,y\}$.
\end{tabbing}
\end{enumerate}
\end{proof}

\begin{proof}[Proof of Proposition \ref{Pprp}] Suppose that
$g_1,\; g_2$ satisfy the mean value exclusion condition. Then the
set of segments $D$ (see Definition \ref{defD}) is the set of
segments promised in the proposition. $D$ satisfies all
the demanded properties:\\
By Lemma \ref{disj} the segments are disjoint.\\
>From the definition of $D$ the segments are open.\\
By Lemma \ref{plmi} all the segments have signs. So we can define\\
a function $G:D\rightarrow \{+1,-1\}$ as follows:\\
 $\forall I \in D$
 \[
   G(I)=\left\{
            \begin{array}{ll}
             +1           & \mbox{if $I$ satisfies the $"+"$ condition}\\
             -1           & \mbox{if $I$ satisfies the $"-"$ condition}\\
            \end{array}
                 \right.
\]
By Lemma \ref{endp} no two segments with a common end point
have the same sign.\\
By Lemmas \ref{idt}, \ref{plmi} and \ref{endp} $g_i$ is a $(D-i,G)$
compatible function, $i=1,2$.
\end{proof}

\subsubsection{Ex post equilibria and parallelograms}

Throughout the proofs we make use the following valuation function
for $i$, $Z_i^{(a,s)} \in {\real^A_+}$, which assigns $a \in A$ the
value $s > 0$ and zero otherwise.

We denote
$g^{(a,a')}_i(s)=b_i(Z_i^{(a,s)}(a))-b_i(Z_1^{(a,s)}(a'))$

\begin{lemma}\label{l3}%
Let $n \ge 2$ and $|A| \ge 3$. Let $b$ be an ex-post equilibrium
for the class of VCG games over $(N,A,{\cal V})$. Assume that for
some $i,j \in N$ and for any $s \in \mathbb{R}$, $Z_i^{(a,s)} \in
{\cal V}_i$ and $Z_j^{(a',s)} \in {\cal V}_j $. Then
$g^{(a,a')}_i$ and $g^{(a',a)}_j$ satisfy the mean exclusion
condition.
\end{lemma}

\begin{proof}
Assume, to the contrary of the claim, $b_i(Z_i^{(a,s)})(a) -
b_i(Z_i^{(a,s)})(a')
> s$, and there exists a player $j$ and some $t$ such that
$Z_j^{(a',t)} \in {\cal V}_j$ and
$$s \ < \ b_j(Z_j^{(a',t)})(a') - b_j(Z_j^{(a',t)})(a) \ \le \
b_i(Z_i^{(a,s)})(a)- b_i(Z_i^{(a,s)})(a'). $$

By Lemma \ref{l1} $b_j(Z_j^{(a',t)})(a')
> b_j(Z_j^{(a',t)})(\hat a)$ for all $\hat a \not =a'$.

Lets consider a simple VCG game with 2 players, $i$ and $j$, where
the mechanism's tie breaking rule, in case of a tie between $a$
and $a'$, is to choose $a$.

Consider the instance where $i$'s valuation is $Z_i^{(a,s)}$ and
$j$'s valuation is $Z_j^{(a',t)}$. Assume that some ${\hat a} \not
\in \{a,a'\}$ is chosen in this game. In this case $i$'s utility
is $j$'s valuation of $\hat a$, namely $b_j(Z_j^{(a',t)})(\hat
a)$. Compare this to $i$'s utility had he announced zero on all
alternatives. In this case $a'$ would have been the chosen
alternative and $i$ would have received a utility of
$b_j(Z_j^{(a',t)})(a')$. As $b_j(Z_j^{(a',t)})(a')
> b_j(Z_j^{(a',t)})(\hat a)$ we have a contradiction with the
assumption that $b$ forms an ex-post equilibrium.

We conclude that either $a$ or $a'$ must chosen.

By our assumption $b_j(Z_j^{(a',t)})(a') + b_i(Z_i^{(a,s)})(a')
\le b_j(Z_j^{(a',t)})(a) + b_i(Z_i^{(a,s)})(a) $, and so $a$ is
actually chosen, and the utility of $i$ is $s +
b_j(Z_j^{(a',t)})(a)$.

On the other hand lets assume $i$ would have announced truthfully.
By the assumption $s+ b_j(Z_j^{(a',t)})(a) <
b_j(Z_j^{(a',t)})(a')$, leading to $a'$ being chosen, and
consequently $i$'s utility would have been
$b_j(Z_j^{(a',t)})(a')$.

By our assumption $b_j(Z_j^{(a',t)})(a') > s +
b_j(Z_j^{(a',t)})(a) $, which stands in contradiction to the fact
the $b$ is an ex-post equilibrium of the 2 player game.
\end{proof}

\begin{corollary}
Let $n \ge 2$ and $|A| \ge 3$. Let $b$ be an ex-post equilibrium
for the class of VCG games over $(N,A,{\cal V})$. For any $a, a'
\in A$ there exists a set of disjoint open segments, denoted
$\Omega^{(a,a')}$ and a function $G^{(a,a')}: \Omega^{(a,a')} \to
\{-1, +1\}$, for which the pair of functions $g^{(a,a')}_1(\cdot)$
and $g^{(a',a)}_2(\cdot)$ are compatible.
\end{corollary}

\begin{proof}
Follows directly from Lemma \ref{l3} and Proposition \ref{Pprp}.
\end{proof}

\begin{lemma}\label{old4}
 Let $n \ge 3$, $|A| \ge 3$, and let $b$ be an ex-post
equilibrium for the class of VCG games over $(N,A,{(\real_+^A)}^n)$,
then $b_i(Z_i^{(a,s)})(a) - b_i(Z_i^{(a,s)})(a') = s$ for all
$i \in N$, $s \in \real_+$ and $a' \not = a \in A$.%
\end{lemma}

\begin{proof}: Assume the claim is not true and that
$b_i(Z_i^{(a,s)})(a) - b_i(Z_i^{(a,s)})(a') \not = s$ for some $i
\in N$, $s \in \real_+$ and $a \in A$. We will assume that
$b_i(Z_i^{(a,s)})(a) - b_i(Z_i^{(a,s)})(a') > s$. The case that
$b_i(Z_i^{(a,s)})(a) - b_i(Z_i^{(a,s)})(a') < s$ is similar, and
therefore omitted.

Choose $t$ such that $b_i(Z_i^{(a,s)})(a) - b_i(Z_i^{(a,s)})(a') > t
> s $ and a player $j \not = i$.

{\em Case 1}: Assume $b_j(Z_j^{(a',t)})(a') - b_j(Z_j^{(a',t)})(a)
\ge t $. If in addition $b_i(Z_i^{(a,s)})(a) - b_i(Z_i^{(a,s)})(a')
\ge b_j(Z_j^{(a',t)})(a') - b_j(Z_j^{(a',t)})(a)$ then we get a
contradiction to lemma \ref{l3}. Otherwise, $b_j(Z_j^{(a',t)})(a') -
b_j(Z_j^{(a',t)})(a)
> b_i(Z_i^{(a,s)})(a) - b_i(Z_i^{(a,s)})(a')$, which leads again to a
contradiction of lemma \ref{l3}, with the roles of $i$ and $j$
reversed.

{\em Case 2}: Assume $b_j(Z_j^{(a',t)})(a') - b_j(Z_j^{(a',t)})(a) <
t $ and consider a third alternative $a'' \not \in \{a,a'\}$. By
lemma \ref{old2} $b_j(Z_j^{(a',t)})(a'') = b_j(Z_j^{(a',t)})(a)$ and
therefore $b_j(Z_j^{(a',t)})(a') - b_j(Z_j^{(a',t)})(a'') < t $.

Consider a third player $l$. Obviously, $b_l(Z_l^{(a'',t)})(a'') -
b_l(Z_l^{(a'',t)})(a) < t $ as well (otherwise we can replicate the
arguments of case 1). By applying lemma \ref{old2} we conclude that
$b_l(Z_l^{(a'',t)})(a'') - b_l(Z_l^{(a'',t)})(a') < t $ as well.
Lets assume, without loss of generality that $b_j(Z_j^{(a',t)})(a')
- b_j(Z_j^{(a',t)})(a'') \le b_l(Z_l^{(a'',t)})(a'') -
b_l(Z_l^{(a'',t)})(a') <t $. This conflicts lemma \ref{l3}, where
$j$ is in the role of $i$ and $l$ in the role of $j$.%
\end{proof}

\begin{lemma} \label{zeroseg}
$I \in \Omega^{(a,a')}$ implies $I^- \not = 0$.
\end{lemma}

\begin{proof}
Assume the claim is wrong. This implies that there exists a player,
w.l.o.g player $1$, and a valuation $v_1$ such that
$v_1(a)-v_1(a')>0$ where $a$ is a maximizing alternative for $v_1$,
but $b_1(v_1)(a)-b_1(v_1)(a')=0$. Among all the VCG mechanisms for
the single player game, there exist one that chooses the alternative
$a'$, in case of tie between $a$ and $a'$. This contradicts the fact
that, in an ex post equilibrium, if player $1$ is on his own that
the maximizing alternative must always be chosen.
\end{proof}

\begin{prp} \label{V1} Let $(b_1,b_2)$ be an ex post equilibrium in the VCG mechanisms.
Let $v_1, v_2 \in \mathcal{V}$ be two valuations for players $1$
and $2$, such that $a$ is a maximizing alternative for $v_1$ and
$a'$ is a maximizing alternative for $v_2$. For any $s$ which is
not an end point of two segments in $\Omega^{(a,a')}$:
\begin{itemize}
\item
If $v_1(a)-v_1(a')=s$  then
$b_1(v_1)(a)-b_1(v)(a')=g^{(a,a')}_1(s)$.
\item
If $v_2(a')-v_2(a)=s$ then
$b_2(v_2)(a')-b_2(v_2)(a)=g^{(a,a')}_2(s)$.
\end{itemize}
\end{prp}

\begin{proof}[Proof of Proposition \ref{V1}]Let $s\in \mathbb{R}_+, v_1\in \mathcal{V}$
such that $v_1(a)-v_1(a')=s$. $a$ is a maximizing alternative for
$v_1$ and $s$ is not an end point of two segments. Then we have:
\begin{enumerate}
\item $g^{(a,a')}_1(s)=s$:
\begin{itemize}
\item If $s=g^{(a,a')}_1(s)<b_1(v_1)(a)-b_1(v_1)(a')$: Choose
$t, s=g^{(a,a')}_1(s)<t<b_1(v_1)(a)-b_1(v_1)(a')$. From the mean
value exclusion condition it follows that $s<g^{(a,a')}_2(t)$.
Consider the profile $(b_1(v_1),b_2(Z_2^{(a',t)}))$. Let $\gamma$
be a maximizing alternative of $(b_1(v_1),b_2(Z_2^{(a',t)}))$. It
follows by Lemma 2 that $\gamma$ is also a maximizing alternative
of $(v_1,b_2(Z_2^{(a',t)}))$. Then $\gamma = a'$, for otherwise
$g_2(t)=b_2(Z_2^{(a',t)}))(a')-b_2(Z_2^{(a',t)}))(\gamma)<v_1(\gamma)-v_1(a')\leq
v_1(a)-v_1(a')=s<g_2(t)$, a contradiction.\\
Again by Lemma 2 $a'$ should be a maximizing alternative of
$(b_1(v_1),Z_2^{(a',t)})$ as well. But
$b_1(v_1)(a)-b_1(v_1)(a')>t=Z_2^{(a',t)}(a')-Z_2^{(a',t)}(a)$, a
contradiction.
\item If $b_1(v_1)(a)-b_1(v_1)(a')<s=g^{(a,a')}_1(s)$: Choose
$t, b_1(v_1)(a)-b_1(v_1)(a')<t<s$. From the mean value exclusion
condition it follows that $g^{(a,a')}_2(t)< s$. Consider the
profile $(b_1(v_1),b_2(Z_2^{(a',t)}))$. Let $\gamma$ be a
maximizing alternative of $(b_1(v_1),b_2(Z_2^{(a',t)}))$. It
follows by Lemma 2 that $\gamma$ is also a maximizing alternative
of $(b_1(v_1),Z_2^{(a',t)})$. Then $\gamma=a'$, for otherwise
$t=Z_2^{(a',t)}(a')-Z_2^{(a',t)}(\gamma)>b_1(v_1)(a)-b_1(v_1)(a')\leq
b_1(v_1)(\gamma)-b_1(v_1)(a')$, a contradiction.\\
Again by Lemma 2 $a'$ should be a maximizing alternative of
$(v_1,b_2(Z_2^{(a',t)}))$ as well. But
$v_1(a)-v_1(a')=s>g^{(a,a')}_2(t)=b_2(Z_2^{(a',t)})(a')-b_2(Z_2^{(a',t)})(a)$,
a contradiction.
\end{itemize}

\item $g^{(a,a')}_1(s)<s$: Let $I\in \Omega^{(a,a')}$ be a segment with $G(I))=-1$
such that $g^{(a,a')}_1(s)=I^-<s\leq I^+$. Such a segment exists by
Proposition \ref{Pprp}. Now consider three cases:
\begin{itemize}
\item $b_1(v_1)(a)-b_1(v_1)(a')<g^{(a,a')}_1(s)<s$:
Choose $t,b_1(v_1)(a)-b_1(v_1)(a')<t<g^{(a,a')}_1(s)$. It emerges
from the mean value exclusion condition that
$g^{(a,a')}_2(t)<g^{(a,a')}_1(s)$. Consider the profile
$(b_1(v_1),b_2(Z_2^{(a',t)}))$. Let $\gamma$ be a maximizing
alternative of $(b_1(v_1),b_2(Z_2^{(a',t)}))$. It follows by Lemma
2 that $\gamma$ is also a maximizing alternative of
$(b_1(v_1),Z_2^{(a',t)})$. Then $\gamma=a'$, for otherwise
$b_1(v_1)(\gamma)-b_1(v_1)(a')\leq b_1(v_1)(a)-b_1(v_1)(a')
<t=Z_2^{(a',t)}(a')-Z_2^{(a',t)}(\gamma)$, a contradiction.\\
Again by Lemma 2 $a'$ should be a maximizing alternative of
$(v_1,b_2(Z_2^{(a',t)}))$ as well. But
$v_1(a)-v_1(a')=s>g^{(a,a')}_1(s)>g^{(a,a')}_2(t))=b_2(Z_2^{(a',t)})(a')-b_2(Z_2^{(a',t)})(a)$,
a contradiction.
\item $g^{(a,a')}_1(s)<b_1(v_1)(a)-b_1(v_1)(a')<s$:
Choose $t, g^{(a,a')}_1(s)<t<b_1(v_1)(a)-b_1(v_1)(a')$. It emerges
from the mean value exclusion condition that $g^{(a,a')}_2(t)\geq
s$. Consider the profile $(b_1(v_1),b_2(Z_2^{(a',t)}))$. Let
$\gamma$ be a maximizing alternative of
$(b_1(v_1),b_2(Z_2^{(a',t)}))$. Then $\gamma=a'$, for otherwise
$b_1(v_1)(\gamma)-b_1(v_1)(a')\leq b_1(v_1)(a)-b_1(v_1)(a')<s\leq
g^{(a,a')}_2(t)=b_2(Z_2^{(a',t)})(a')-b_2(Z_2^{(a',t)})(\gamma)$
, a contradiction.\\
Again by Lemma 2 $a'$ should be a maximizing alternative of
$(b_1(v_1),Z_2^{(a',t)})$ as well. But
$b_1(v_1)(a)-b_1(v_1)(a')>t=Z_2^{(a',t)}(a')-Z_2^{(a',t)}(a)$, a
contradiction.
\item $g^{(a,a')}_1(s)<s\leq b_1(v_1)(a)$: There are three
cases:\\
\begin{enumerate}
\item $s=b_1(v_1)(a)$: Choose $t, g^{(a,a')}_1(s)<t<s$.
It emerges from the mean value exclusion condition that
$g^{(a,a')}_2(t)\geq s$. Consider the profile
$(b_1(v_1),b_2(Z_2^{(a',t)}))$. Note that $a'$ is a maximizing
alternative of $(b_1(v_1),b_2(Z_2^{(a',t)}))$ (not necessarily the
only one). For otherwise, there exists an alternative $\gamma$
which gives a better social surplus. But,
$b_1(v_1)(\gamma)-b_1(v_1)(a')\leq b_1(v_1)(a)-b_1(v_1)(a')\leq
g^{(a,a')}_2(t)=b_2(Z_2^{(a',t)})(a')-b_2(Z_2^{(a',t)})(\gamma)$,
a contradiction.\\
By Lemma 2 $a'$ should be a maximizing alternative of
$(b_1(v_1),Z_2^{(a',t)})$ as well. But
$b_1(v_1)(a)-b_1(v_1)(a')>t=Z_2^{(a',t)}(a')-Z_2^{(a',t)}(a)$, a
contradiction.
\item $s=I^+, s<b_1(v_1)(a)-b_1(v_1)(a')$: There are two cases induced when
$s$ is not an end point of two segments:
\begin{enumerate}
\item $I^+$ is a limit point of segments $I_k\in \Omega^{(a,a')}$ that lie to
the right of $I^+$. Then we can find a segment $I_j$ such that
$I^+<I_j^-<I_j^+<b_1(v_1)(a)-b_1(v_1)(a')$. Choosing a number
$s_0\in I_j$ we have
$I^+<g^{(a,a')}_1(s_0),g^{(a,a')}_2(s_0)<b_1(v_1)(a)-b_1(v_1)(a')$.
\item $I^+$ is not a limit of segments. Then we
can find a number $s_0, s<s_0<b_1(v_1)(a)$ such that $s_0\in
\mathbb{R}_+ \backslash \displaystyle{\bigcup_{I
\in{\Omega^{(a,a')}}}}{\bar{I}}$ where $\bar{I}$ is the closure of
$I$. For such $s_0$ we have
$I^+<g^{(a,a')}_1(s_0)=g^{(a,a')}_2(s_0)=s_0<b_1(v_1)(a)$.
\end{enumerate}
In both cases we shall look at the profile
$(b_1(v_1),b_2(Z_2^{(a',s_0)}))$. As
$g^{(a,a')}_2(s_0)<b_1(v_1)(a)-b_1(v_1)(a')$, it follows that $a$ is
a maximizing alternative. By Lemma 2 it should be a maximizing
alternative of $(v_1,b_2(Z_2^{(a',s_0)}))$ as well. But
$v_1(a)-v_1(a')=s=I^+<g^{(a,a')}_2(s_0)=b_2(Z_2^{(a',s_0)})(a')-b_2(Z_2^{(a',s_0)})(a)$,
a contradiction.
\item $s<b_1(v_1)(a)-b_1(v_1)(a'), s\neq I^+$: Then $s<I^+=g^{(a,a')}_2(I^+)$ and $\forall s', s<s'<I^+$,
we have $g^{(a,a')}_2(s')=I^+$. There are two cases:
\begin{enumerate}
\item $b_1(v_1)(a)-b_1(v_1)(a')\leq I^+$. Choose $s', s<s'<b_1(v_1)(a)-b_1(v_1)(a')\leq I^+$.
Consider the profile $(b_1(v_1),b_2(Z_2^{(a',s')}))$.\\
Note that $a'$ is a maximizing alternative of
$(b_1(v_1),b_2(Z_2^{(a',s')}))$ (not necessarily the only one). For
otherwise, there exists an alternative $\gamma$ which gives a better
social surplus. But, $b_1(v_1)(\gamma)-b_1(v_1)(a')\leq
b_1(v_1)(a)-b_1(v_1)(a')\leq
I^+=b_2(Z_2^{(a',s')})(a')-b_2(Z_2^{(a',s')})(\gamma)$, a
contradiction. It follows by Lemma 2 that $a'$ also maximizes
$(b_1(v_1),Z_2^{(a',s')})$. But
$b_1(v_1)(a)-b_1(v_1)(a')>s'=Z_2^{(a',s')}(a')-Z_2^{(a',s')}(a)$,
a contradiction.\\
\item $I^+<b_1(v_1)(a)-b_1(v_1)(a')$. Consider the
profile $(b_1(v_1),b_2(Z_2^{(a',I^+)}))$. Note that $a$ is a
maximizing alternative of $(b_1(v_1),b_2(Z_2^{(a',I^+)}))$ (not
necessarily the only one). For otherwise, there exists an
alternative $\gamma$ which gives a better social surplus. But if
$\gamma=a'$ then, $b_1(v_1)(a)-b_1(v_1)(a')> I^+
=b_2(Z_2^{(a',I^+)})(a')-b_2(Z_2^{(a',I^+)})(a)$, a contradiction.
If $\gamma \neq a,a'$ then, $b_1(v_1)(a)-b_1(v_1)(\gamma)\geq 0 =
b_2(Z_2^{(a',I^+)})(\gamma)-b_2(Z_2^{(a',I^+)})(a)$, a
contradiction. It follows by Lemma 2 that it also maximizes
$(v_1,b_2(Z_2^{(a',I^+)}))$. But $v_1(a)-v_1(a')=s<I^+=
b_2(Z_2^{(a',I^+)})(a')-b_2(Z_2^{(a',I^+)})(a)$, a
contradiction.\\
\end{enumerate}
\end{enumerate}
\end{itemize}
\item $g^{(a,a')}_1(s)>s$: This case is handled in a similar way as the previous
case.
\end{enumerate}
For player 2 the proof is similar.
\end{proof}

\begin{prp} \label{V2} Let $(b_1,b_2)$ be an ex post equilibrium in the VCG mechanisms.
Let $v_1, v_2 \in \mathcal{V}$ be two valuations for players $1$
and $2$, such that $a$ is a maximizing alternative for $v_1$ and
$a'$ is a maximizing alternative for $v_2$. Let $s$ be an end
point of two segments $I=(x,s), J=(s,y) \in \Omega^{(a,a')}$. If
$v_1(a)-v_1(a')=s$ and $v_2(a')-v_2(a)=s$ one of the following
must hold:
\begin{enumerate}
\item
If $G^{(a,a')}(I) = -1$ and $G^{(a,a')}(J) = +1$ then:\\
$b_1(v_1)(a)-b_1(v_1)(a')=x\; or\; y$\\
$b_2(v_2)(a')-b_2(v_2)(a)=s$
\item
If $G^{(a,a')}(I) = +1$ and $G^{(a,a')}(J) = -1$ then:\\
$b_1(v_1)(a)-b_1(v_1)(a')=s$\\
$b_2(v_2)(a')-b_2(v_2)(a)=x\; or\; y$

\end{enumerate}
\end{prp}

\begin{proof} [Proof of Proposition \ref{V2}]
Let $s\in \mathbb{R}_+, v_1\in \mathcal{V}$
such that $v_1(a)-v_1(a')=s$. $a$ is a maximizing alternative for
$v_1$ and $s$ is an end point of two segments $I=(x,s), J=(s,y)
\in \Omega^{(a,a')}$. There are three cases to consider, $g^{(a,a')}_1(s)=s$, $g^{(a,a')}_1(s)<s$ and $g^{(a,a')}_1(s)>s$:

\begin{enumerate}
\item $g^{(a,a')}_1(s)=s$: In the proof of \ref{V1} where
$g^{(a,a')}_1(s)=s$, there was no use of the fact that $s$ wasn't
an end point of two segments. Therefore the result is valid in
this case as well.
\item $g^{(a,a')}_1(s)<s$:
%(see Figure 1 below)
%\begin{center}
%\begin{figure}[h]\label{mak3}

%%\scalebox{1}{\includegraphics{mak3ns}}

%\caption{ An illustration of a parallelogram scheme in which $s$ is
%an end point of two segments and $g_1(s)<s$}\label{mak3ns.eps}
%\end{figure}
%\end{center}

Assume for the sake of contradiction that
$b_1(v_1)(a)-b_1(v_1)(a') \notin \{g^{(a,a')}_1(s),y\}$. There are
four cases:
\begin{enumerate}
\item $b_1(v_1)(a)-b_1(v_1)(a')<g^{(a,a')}_1(s)=x$: Choose $t,
b_1(v_1)(a)-b_1(v_1)(a')<t<g^{(a,a')}_1(s)=x$. Consider the
profile $(b_1(v_1),b_2(Z_2^{(a',t)}))$. Let $\gamma$ be a
maximizing alternative for this profile. Then by Lemma 2 it
maximizes $(b_1(v_1),Z_2^{(a',t)})$ as well. Hence $\gamma=a'$,
for otherwise $b_1(v_1)(\gamma)-b_1(v_1)(a')\leq
b_1(v_1)(a)-b_1(v_1)(a')< t =
Z_2^{(a',t)}(a')-Z_2^{(a',t)}(\gamma)$. It also follows by Lemma 2
that $a'$ maximizes $(v_1,b_2(Z_2^{(a',t)}))$. Therefore $
g^{(a,a')}_2(t)= b_2(Z_2^{(a',t)})(a')-b_2(Z_2^{(a',t)})(a)\geq
v_1(a)-v_1(a')=s$. This contradicts the mean value exclusion
condition.
\item $g^{(a,a')}_1(s)<b_1(v_1)(a)-b_1(v_1)(a')\leq s$: Choose $t,
g^{(a,a')}_1(s)<t<b_1(v_1)(a)-b_1(v_1)(a')\leq s$. Then
$g^{(a,a')}_2(t)=s.$ Consider the profile
$(b_1(v_1),b_2(Z_2^{(a',t)}))$. Note that $a'$ is a maximizing
alternative of $(b_1(v_1),b_2(Z_2^{(a',t)}))$ (not necessarily the
only one). For otherwise, there exists an alternative $\gamma$
which gives a better social surplus. But,
$b_1(v_1)(\gamma)-b_1(v_1)(a')\leq b_1(v_1)(a)-b_1(v_1)(a')\leq s=
g^{(a,a')}_2(t)=b_2(Z_2^{(a',t)})(a')-b_2(Z_2^{(a',t)})(\gamma)$,
a contradiction.\\
By Lemma 2 it maximizes $(b_1(v_1),Z_2^{(a',t)})$ as well, but
$b_1(v_1)(a)-b_1(v_1)(a')>t=Z_2^{(a',t)}(a')-Z_2^{(a',t)}(a)$, a
contradiction.
\item $s<b_1(v_1)(a)-b_1(v_1)(a')<y$: Choose $t,
s<b_1(v_1)(a)-b_1(v_1)(a')<t<b$. Then $g^{(a,a')}_2(t)=s.$
Consider the profile $(b_1(v_1),b_2(Z_2^{(a',t)}))$. Note that $a$
is a maximizing alternative of $(b_1(v_1),b_2(Z_2^{(a',t)}))$ (not
necessarily the only one). For otherwise, there exists an
alternative $\gamma$ which gives a better social surplus. But if
$\gamma=a'$ then, $b_1(v_1)(a)-b_1(v_1)(a')> s
=g^{(a,a')}_2(t)=b_2(Z_2^{(a',t)})(a')-b_2(Z_2^{(a',t)})(a)$, a
contradiction. If $\gamma \neq a,a'$ then,
$b_1(v_1)(a)-b_1(v_1)(\gamma)\geq 0 =
b_2(Z_2^{(a',t)})(\gamma)-b_2(Z_2^{(a',t)})(a)$,
a contradiction.\\
By Lemma 2 it maximizes $(b_1(v_1),Z_2^{(a',t)})$ as well, but
$b_1(v_1)(a)-b_1(v_1)(a')<t=Z_2^{(a',t)}(a')-Z_2^{(a',t)}(a)$, a
contradiction.
\item $y<b_1(v_1)(a)-b_1(v_1)(a')$: Choose $t,b<t<b_1(v_1)(a)$. Consider the profile
$(b_1(v_1),b_2(Z_2^{(a',t)}))$. Let $\gamma$ be a maximizing
alternative for this profile. Then by Lemma 2 it maximizes
$(b_1(v_1),Z_2^{(a',t)})$ as well.
Hence $\gamma\neq a'$, for otherwise\\
$b_1(v_1)(a)-b_1(v_1)(a')>t=Z_2^{(a',t)}(a')-Z_2^{(a',t)}(a)$,
a contradiction.\\
It also follows by Lemma 2 that $\gamma$ maximizes
$(v_1,b_2(Z_2^{(a',t)}))$. Therefore
$g^{(a,a')}_2(t)=b_2(Z_2^{(a',t)})(a')-b_2(Z_2^{(a',t)})(\gamma)=
b_2(Z_2^{(a',t)})(a')-b_2(Z_2^{(a',t)})(a)<
v_1(\gamma)-v_1(a')<v_1(a)-v_1(a')=s$. This contradicts the mean
value exclusion condition.
\end{enumerate}
\item $g^{(a,a')}_1(s)>s$: This case is handled in a similar way as
the previous case.
\end{enumerate}
For player 2 the proof is similar.
\end{proof}

\begin{lemma}\label{lemRemark1}
Let $a,a'$ be any two alternatives and let $\Omega^{(a,a')}$ be the
set of segments induced by Proposition \ref{Pprp} then for any
$\epsilon>0$ we can find a $\delta, 0<\delta<\epsilon$ such that
$0<g_1^{(a,a')}(\delta),g_2^{(a,a')}(\delta)<\epsilon$.
\end{lemma}

\begin{proof}[Proof of Lemma \ref{lemRemark1}]
Let $\epsilon>0$, if there exists a segment $I=(x,y)$ in
$\Omega^{(a,a')}$ such that $y<\epsilon$ then due to Lemma
\ref{zeroseg} $x>0$, and from Proposition \ref{Pprp} it follows that
for any $\delta$ such that $x<\delta<y$ , $0<x\leq
g_1^{(a,a')}(\delta),g_2^{(a,a')}(\delta)\leq y<\epsilon$. Other
wise there are two cases left:\\
case 1: there exists a segment $I=(x,y)$ in $\Omega^{(a,a')}$ such
that $x<\epsilon<y$ again due to Lemma \ref{zeroseg} $0<x$. As for
this cases conditions it follows that for any segment $J\in
\Omega^{(a,a')}, J\bigcap (0,x)=\phi$. So, for any $\delta \in
(0,x),  g_1^{(a,a')}(\delta)=g_2^{(a,a')}(\delta)=\delta$ where
$0<\delta<x<\epsilon$ as required.\\
case 2: For any segment $J\in \Omega^{(a,a')}, j\bigcap
(0,\epsilon)=\phi$. then for any $\delta \in (0,\epsilon),
g_1^{(a,a')}(\delta)=g_2^{(a,a')}(\delta)=\delta$ where
$0<\delta<\epsilon$ as required.\\
\end{proof}

\subsection{Proof of Theorems \ref{NoPar1} and \ref{NoPar2}}

\subsubsection{The Easy Direction}

\begin{proof}

{\bf The easy direction of Theorem \ref{NoPar1}}: We shall first
show for an arbitrary set of functions $f_i:\mathcal{V}\rightarrow
\mathbb{R_+}, i=1,...,n$. The strategy tupple
$b_i(v_i)(a)=v_i(a)+f_i(v_i)$ forms an ex-post equilibrium for the
class of VCG games over $(N,A,(\mathbb{R_+}^A)^n)$.

Consider a VC mechanism $d$, a profile of valuations
$v=(v_1,...,v_n)\in V^N$, $n$ arbitrary functions
$f_i:\mathcal{V}\rightarrow \mathbb{R_+}$ and a buyer i. According
to the strategies $b_i(v_i)(a)=v_i(a)+f_i(v_i)$, the profile of
announced valuations is
$\hat{v}=(v_1(a)+f_1(v_1),...,v_n(a)+f_n(v_n))$. Let
\[
t=\max_{\hat{a}\in A}\sum_{j\neq i}v_j(\hat{a})+f_j(v_j).
\]
Let $v'$ be the profile of announced valuations consisting of an
arbitrary announcement $v_i'$ of buyer $i$ and the fixed
announcements $v_j(a)+f_j(v_j)$ of buyers $j\in N\backslash
\{i\}$. Suppose that the alternative $d(v')$ is $a'$. Then the
utility of buyer $i$ is
\[
u_i^d(v_i,v')=v_i(a')-c_i^d(v)=v_i(a')+\sum_{j\neq
i}v_j(a')+f_j(v_j)-t.
\]
This is maximized when $a'$ maximizes $v_i(\hat{a})+\sum_{j\neq
i}v_j(\hat{a})$. But this is exactly what the mechanism maximizes
when it chooses an alternative. So, by announcing $v_i$ the
utility of $i$ will be maximized. But if he announces
$v_i+f_i(v_i)$ where $f_i(v_i)$ does not change on the different
alternatives then the mechanism still maximizes $i$'s utility.

Note that the above arguments fully mimic the proof of the standards
arguments for proving that VCG mechanisms are incentive compatible.

{\bf The easy direction of Theorem \ref{NoPar2}}: Follows as a
corollary from the above arguments and Proposition \ref{prop1}.

\subsubsection{The Difficult Direction: {\bf $n\ge 2$} and {\bf $A>2$}}

{\bf The difficult direction of Theorem \ref{NoPar2}}: In fact, to
prove this direction we may assume, with out loss of generality,
that there are exactly $n=2$ players (recall the definition of an
ex post
equilibrium) and $|A|\ge 3$, or, alternatively that there are $n=3$ players.%
\footnote{Indeed, suppose there more players and there is an ex
post equilibrium which is not of the form stated in the Theorem.
By definition, it must be an equilibrium for 2 players as
well.}

Assume for the sake of contradiction that the claim is wrong and
that for some ex-post equilibrium b, there exists an agent $i$,
without loss of generality  $i=1$ and a valuation function ,$v_1$,
and two alternatives, $a,a' \in A$ such $b_1(v_1)(a)-v_1(a)\neq
b_1(v_1)(a')-v_1(a')$. Without loss of generality we may choose
$a$ such that $v_1(a)=argmax_{\hat{a}\in A}v_1(\hat{a})$. There
are two cases:
\begin{enumerate}
\item $b_1(v_1)(a)-v_1(a)>b_1(v_1)(a')-v_1(a')$.
In this case $b_1(v_1)(a)-b_1(v_1)(a')>v_1(a)-v_1(a')$ by
Proposition \ref{Pprp}, Proposition \ref{V1} and Proposition
\ref{V2} the corresponding $\Omega^{(a,a')}$ is not empty and there
exists a segment $I$ such that $G(I)=+1$ in $\Omega^{(a,a')}$.
Denote $I=(x_1,x_2)$ and
$h=x_2-x_1$.\\
Consider the following two valuations:
\[
    u_2(\hat{a})=\left\{
            \begin{array}{ll}
             x_2-h_3-x_1         & \mbox{if $\hat{a}=a$}\\
             x_2                 & \mbox{if $\hat{a}=a'$}\\
             x_2-h_2-x_1         & \mbox{if $\hat{a}=a''$}\\
             x_2-x_1-h_4         & \mbox{otherwise}
            \end{array}
                 \right.
\]
Where $0<h_2<h_3<h_4<h$.
\[
    u_1(\hat{a})=\left\{
            \begin{array}{ll}
             M                   & \mbox{if $\hat{a}=a$}\\
             M-x_1-h_1           & \mbox{if $\hat{a}=a'$}\\
             M-x_1+\delta      & \mbox{if $\hat{a}=a''$}\\
             0                   & \mbox{otherwise}
            \end{array}
                 \right.
\]
Where $a'' \neq a,a'$ $x_2<M$, $0<h_1<h$ and $\delta$ is chosen such
that it is not a common end point of two segments in
$\Omega^{(a,a'')}$ and $0<g_1^{(a,a'')}(x_1-\delta)<h_3-h_2$,
$0<\delta<x_1$ this is possible as for lemma \ref{lemRemark1}.

Note that the following emerges from Proposition \ref{Pprp} and
Proposition \ref{V1}:\\
$b_1(u_1)(a)-b_1(u_1)(a')=x_2$\\
$b_1(u_1)(a)-b_1(u_1)(a'')=g_1^{(a,a'')}(x_1-\delta)$\\

$b_2(u_2)(a')-b_2(u_2)(a)=x_1$\\
$b_2(u_2)(a')-b_2(u_2)(a'')=x_1$\\

The following will show that for the profile of strategies
$(b_1(u_1),b_2(u_2))$, the alternative $a$ is the only maximizing alternative.\\

We show that the total announcements at $a$ exceeds that of $\hat{a}\in A$, where $\hat{a}\neq a,a',a''$:\\
$b_1(u_1)(a)-b_1(u_1)(\hat{a})=g_1^{(a,\hat{a})}(M)=
g_1^{(a,a')}(M)\geq x_2$\\
the second equality follows from lemma \ref{old2}.\\
$b_2(u_2)(a')-b_2(u_2)(a)=x_1$\\
$b_2(u_2)(a')-b_2(u_2)(\hat{a})=x_1$ again from lemma \ref{old2}.\\
Then it follows that $b_2(u_2)(\hat{a})-b_2(u_2)(a)=0$\\
So we have that $b_2(u_2)(\hat{a})-b_2(u_2)(a)=0<x_2\leq
b_1(u_1)(a)-b_1(u_1)(\hat{a})$\\
which means that $b_1(u_1)(\hat{a})+b_2(u_2)(\hat{a})<
b_1(u_1)(a)+b_2(u_2)(a)$.\\

We now show that the total announcements at $a$ exceeds that of $a'$:\\
$b_2(u_2)(a')-b_2(u_2)(a)=x_1<x_2=b_1(u_1)(a)-b_1(u_1)(a')$.\\

We now show that The total announcements at $a$ exceeds that of $a''$:\\
$b_1(u_1)(a)-b_1(u_1)(a'')=g_1^{(a,a'')}(x_1-\delta)$\\
$b_2(u_2)(a')-b_2(u_2)(a)=x_1$\\
$b_2(u_2)(a')-b_2(u_2)(a'')=x_1$\\
it follows that $b_2(u_2)(a'')-b_2(u_2)(a)=0$\\
So we that
$b_2(u_2)(a'')-b_2(u_2)(a)=0<g_1^{(a,a'')}(x_1-\delta)=\\
b_1(u_1)(a)-b_1(u_1)(a'')$.\\

This proves that $a$ is the only maximum of $(b_1(u_1),b_2(u_2))$.
By lemma \ref{old5} $a$ should be a maximum of the profile
$(b_1(u_1),u_2)$, but
$u_2(a'')-u_2(a)=x_2-h_2-x_1-x_2+h_3+x_1=h_3-h_2>g_1^{(a,a'')}(x_1-\delta)=b_1(u_1)(a)-b_1(u_1)(a'')$
which means that $u_2(a'')+b_1(u_1)(a'')>b_1(u_1)(a)+u_2(a)$ a
contradiction.
\item The proof for the case that $b_1(v_1)(a)-v_1(a)<b_1(v_1)(a')-v_1(a')$
is similar to the previous case, and is therefore omitted.
\end{enumerate}

{\bf The difficult direction of Theorem \ref{NoPar1}}: Follows as a
corollary from the proof of the difficult direction of Theorem
\ref{NoPar2} and Proposition \ref{prop1}.

\end{proof}

\subsubsection{The Difficult Direction: {\bf $n\ge 3$}}

\begin{proof}[Proof of Theorem \ref{NoPar2}]: Assume the claim is wrong and that
for some ex-post equilibrium $b$, there exists an agent $i$ and a
valuation function, $v_i$, and two alternatives, $a,a' \in A$ such
$b_i(v_i)(a) - v_i(a) \not = b_i(v_i)(a') - v_i(a')$. Without loss
of generality we may choose $a$ such that $b_i(v_i)(a) =
\arg\max_{{\hat a} \in A} b_i(v_i)({\hat a})$.

{\em Case 1}: $b_i(v_i)(a) - v_i(a) > b_i(v_i)(a') - v_i(a')$.

In this case $b_i(v_i)(a) + v_i(a')> b_i(v_i)(a') + v_i(a)$. Let $t$
satisfy $b_i(v_i)(a) - b_i(v_i)(a') > t >  v_i(a) - v_i(a')$, and
consider player $j$ and a valuation $v_j = Z_j^{(a',t)}$.

By lemma \ref{old2} $b_j(Z_j^{(a',t)})(a) = b_j(Z_j^{(a',t)})(a'')$,
which together with the choice of $a$ implies $b_i(v_i)(a) +
b_j(Z_j^{(a',t)})(a) \ge b_i(v_i)(a'') + b_j(Z_j^{(a',t)})(a'') $.
By a proper choice of the mechanism we may find a simple VCG game
such that $a''$ is not chosen, and therefore, either $a$ or $a'$ are
chosen.

Assume $a$ is chosen. Then the utility of $i$ is $v_i(a) +
b_j(Z_j^{(a',t)})(a)$, which, by lemma \ref{old4} is equal $v_i(a) +
b_j(Z_j^{(a',t)})(a') - t$. This in turn is less than $v_i(a) +
b_j(Z_j^{(a',t)})(a') - (v_i(a) - v_i(a')) = b_j(Z_j^{(a',t)})(a') +
v_i(a')$, contradicting lemma \ref{old5}.

Therefore, it must be the case that $a'$ is chosen. However,
consider $j$'s utility, $t + b_i(v_i)(a') < b_i(v_i)(a') -
b_i(v_i)(a') + b_i(v_i)(a) = b_i(v_i)(a)$, again contradicting lemma
\ref{old5}.

{\em Case 2}: $b_i(v_i)(a) - v_i(a) < b_i(v_i)(a') - v_i(a')$.

This case is repeated with analogous arguments with $v_j =
Z_j^{(a',t)}$, where $b_i(v_i)(a) - b_i(v_i)(a') < t <  v_i(a) -
v_i(a')$.

\end{proof}

\begin{proof}[Proof of Theorem \ref{NoPar1}]: Follows as a
corollary from the proof of the difficult direction of Theorem
\ref{NoPar2} and Proposition \ref{prop1}.
\end{proof}

\subsection{Proof of Theorems \ref{thm3} and \ref{thm4}}

\begin{proof} [Proof of Theorem \ref{thm3}]: Assume that for some $N' \subset N$
and for some specific realization of valuations, there exists a
player $i \in N'$ which can benefit from deviation. This means that
deviating to the strategy ${\hat b}_i(v_i)(a) = v_i(a)$ is also
strictly beneficial (recall that truth telling is a dominant
strategy for all VCG games). However, truth telling cannot change
the chosen alternative and therefore cannot change $i$'s utility.
\end{proof}

Throughout this subsection we fix the valuation sets, ${\cal
R}_i(a_i)$. For each player $i$, let $a_i$ denote the optimal
element with respect to ${\cal R}_i(a_i)$. The following lemma is in
the spirit of lemma \ref{old2}.

\begin{lemma}\label{l6} Let $n \ge 3$ and let $b$ be an ex-post equilibrium for the
class of VCG games over $(N,A,{\cal R}(a))$. Then for all $i$, $k$,
$m$ and all
$Z^{(a_i,t)}_i$, $b_i(Z^{(a_i,t)}_i)(a_k) = b_i(Z^{(a_i,t)}_i)(a_m)$ for all $a_k,a_m \not = a_i$.%
\end{lemma}

\begin{proof} Assume the claim is wrong, and that for some $i$, $k$,
$m$ and $t$, $b_i(Z^{(a_i,t)}_i)(a_k) > b_i(Z^{(a_i,t)}_i)(a_m)$,
where $a_i$, $a_k$ and $a_m$ are three distinct alternatives. Now we
can mimic the arguments in the proof of lemma \ref{old2} and reach a
contradiction.
\end{proof}

One can note that the proof above provides a slightly stronger
result, for which we only need $2$ players:

\begin{lemma}\label{l7} Let $n \ge 2$ and let $b$ be an ex-post equilibrium for the
class of VCG games over $(N,A,{\cal R}(a))$. Then for all $i$, $k$
and all $Z^{(a_i,t)}_i$, $b_i(Z^{(a_i,t)}_i)(a_k) \ge
b_i(Z^{(a_i,t)}_i)({\hat a})$ for all ${\hat a} \not = a_i$.
\end{lemma}

The next lemma is quite similar to lemma \ref{old4}, and its proof
is identical and therefore omitted:

\begin{lemma}\label{l8} Let $n \ge 3$, $|A| \ge 3$, and let $b$ be an ex-post
equilibrium for the class of VCG games over $(N,A,{\cal R}(a))$,
then $b_i(Z_i^{(a_i,s)})(a_i) - b_i(Z_i^{(a_i,s)})(a_k) = s$ for all
$i \in N$, $s \in \mathbb{R}_+$ and $a_k \not = a_i \in A$.%
\end{lemma}

\begin{proof}[Proof of Theorem \ref{thm4}]: Follows the arguments in the proof of
Theorem \ref{NoPar2} for the case $n \ge 3$, where the reference to
lemmas \ref{old2} and \ref{old4} are replaced with lemma \ref{l6}
and \ref{l8}.
\end{proof}


\begin{thebibliography}{10}

\bibitem{Anderson}
A.~Anderson, M.~Tenhunen, and F.~Ygge, \emph{Integer programming for
  combinatorial auction winner determination}, ICMAS, 2000, pp.~39--46.

\bibitem{BergemannMorris2008}
D.~Bergemann and S.~Morris, \emph{Ex post implementation}, Games and Economic
  Behavior \textbf{63} (2008), 527--566.

\bibitem{Bikhchandani2006}
S.~Bikhchandani, \emph{Ex post implementation in environments with private
  goods}, Theoretical Economics \textbf{1} (2006), 369--393.

\bibitem{Bikhchandani2001}
S.~Bikhchandani, S.~de~Vries, J.~Schummer, and R.~Vohra, \emph{Linear
  programming and vickrey auctions}, Mathematics of the Internet: E-auction and
  Markets (Dietrich and Vohra, eds.), Springer, New York, 2002, pp.~413--424.

\bibitem{FujiBrownShoham}
Y.~Fujishima, K.~Leyton-Brown, and Y.~Shoham, \emph{Taming the computational
  complexity of combinatorial auctions: Optimal and approximate approaches},
  IJCAI-99, 1999.

\bibitem{GulSta00}
F.~Gul and E.~Stacchetti, \emph{The english auction with differentiated
  commodities}, Journal of Economic Theory \textbf{92} (2000), no.~1, 66--95.

\bibitem{RNDM}
R.~Holzman, N.~Kfir-Dahav, D.~Monderer, and M.~Tennenholtz, \emph{{Bundling
  equilibrium in combinatorial auctions}}, Games and Economic Behavior
  \textbf{47} (2004).

\bibitem{HolMon2002}
R.~Holzman and D.~Monderer, \emph{{Characterization of ex post Equilibrium in
  the VCG Combinatorial Auctions}}, Games and Economic Behavior \textbf{47}
  (2004).

\bibitem{HoosBoutilier}
H.H. Hoos and C.~Boutilier, \emph{Solving combinatorial auctions using
  stochastic local search}, The 17th national conference on artificial
  intelligence, 2000, pp.~22--29.

\bibitem{JMMZ}
P.~Jehiel, M.~Meyer ter Vehn, B.~Moldovanu, and W.R. Zame, \emph{The limits of
  ex post implementation}, Econometrica \textbf{74(3)} (2005), 585--610.

\bibitem{NisSeg02}
N.~Nisan and I.~Segal, \emph{The communication requirements of efficient
  allocations and supporting prices}, Journal of Economic Theory \textbf{129}
  (2006), 192--224.

\bibitem{Parkes99}
D.~C. Parkes, \emph{ibundle: An efficient ascending price bundle auction}, ACM
  Conference on Electronic Commerce, 1999.

\bibitem{ParkesUngar}
D.~C. Parkes and L.~H. Ungar, \emph{Iterative combinatorial auctions: Theory
  and practice}, AAAI,IAAI, 2000, pp.~74--81.

\bibitem{RothPekHar}
M.H. Rothkopf, A.~Pekec, and R.M. Harstad, \emph{Computationally manageable
  combinatorial auctions}, Management Science \textbf{44} (1998), no.~8,
  1131--1147.

\bibitem{Sandholm01}
T.~Sandholm, S.~Suri, A.~Gilpin, and D.~Levine, \emph{Cabob: A fast optimal
  algorithm for combinatorial auctions}, 17th International Joint Conference on
  Artificial Intelligence, 2001, pp.~1102--1108.

\bibitem{Wellmangeb}
M.P. Wellman, P.R. Wurman, W.E. Walsh, and J.K. MacKie-Mason, \emph{Auction
  protocols for decentralized scheduling}, Games and Economic Behavior
  \textbf{35} (2001), 271--303.

\end{thebibliography}
\end{document}